\documentclass[10pt, conference, letterpaper]{IEEEtran}
\usepackage{multirow}
\IEEEoverridecommandlockouts
\usepackage{acronym}
\usepackage[left=1.57cm,right=1.57cm,top=1.5cm,bottom=2cm]{geometry}
\usepackage{amsthm}
\usepackage{cite}
\usepackage{amsmath,amssymb,amsfonts}
\usepackage{graphicx}
\usepackage{textcomp}
\usepackage{xcolor}
\def\BibTeX{{\rm B\kern-.05em{\sc i\kern-.025em b}\kern-.08em
    T\kern-.1667em\lower.7ex\hbox{E}\kern-.125emX}}

\usepackage{url}
\usepackage{bbm}
\usepackage{subcaption}
\usepackage{tikz}

\newtheorem{problem}{Problem}
\newtheorem{corollary}{Corollary}
\usepackage[utf8]{inputenc}
\usepackage{algorithm}
\usepackage[noend]{algpseudocode}
\algnewcommand{\Initialize}[1]{%
  \State \textbf{Initialize:}
  \Statex \hspace*{\algorithmicindent}\parbox[t]{.8\linewidth}{\raggedright #1}
}

\DeclareMathOperator*\argmax{arg \, max}		

\DeclareMathOperator*\minimize{min.}
\DeclareMathOperator*\subject{s.t.}	

\newcommand{\field}[1]{\mathbb{#1}}

\newcommand{\set}[1]{\mathcal{#1}}

\newcommand{\R}{{\field{R}}}
\newcommand{\C}{{\field{C}}}
\newcommand{\N}{{\field{N}}}
\newcommand{\Ex}{{\field{E}}}

\newcommand{\ma}[1]{\boldsymbol{\mathbf{#1}}}
\newcommand{\ve}[1]{\boldsymbol{\mathbf{#1}}}

\newcommand{\mI}{\ma{I}}

\newcommand{\vs}{\ve{s}}
\newcommand{\vm}{\ve{m}}
\newcommand{\vp}{\ve{p}}

\newcommand{\vo}{\ve{o}}
\newcommand{\vq}{\ve{q}}

\newcommand{\vx}{\ve{x}}
\newcommand{\vy}{\ve{y}}
\newcommand{\vv}{\ve{v}}
\newcommand{\vu}{\ve{u}}
\newcommand{\vz}{\ve{z}}

\newcommand{\vmu}{\ve{\mu}}
\newcommand{\msig}{\ma{\Sigma}}

\newcommand{\Ns}{\set{N}}
\newcommand{\Xs}{\set{X}}
\newcommand{\Ss}{\set{S}}
\newcommand{\Ks}{\set{K}}
\newcommand{\Ls}{\set{L}}
\newcommand{\Ds}{\set{D}}
\newcommand{\KL}{\set{D}_{\Ks\Ls}}

\newcommand{\Ys}{\set{Y}}

\newcommand{\operator}[1]{\mathrm{#1}}

\newcommand{\thr}{(\operator{th})}
\newcommand{\app}{(\operator{APP})}
\newcommand{\ib}{(\operator{IB})}
\newcommand{\vib}{(\operator{VIB})}

\begin{document}

\title{AdaSem: Adaptive Goal-Oriented Semantic Communications for End-to-End Camera Relocalization}

\author{
\IEEEauthorblockN{Qi Liao\IEEEauthorrefmark{1}, Tze-Yang Tung\IEEEauthorrefmark{2}}
\IEEEauthorblockA{ 
\IEEEauthorrefmark{1}Nokia Bell Labs, Stuttgart, Germany\\
\IEEEauthorrefmark{2}Nokia Bell Labs, Murray Hill, United States\\
Email: \url{{qi.liao, tze-yang.tung}@nokia-bell-labs.com}}
}

\maketitle
\def\thefootnote{}\footnotetext{The authors contributed equally to this work. This work has been accepted for IEEE International Conference on Computer Communications (INFOCOM) 2024.}
\begin{abstract}
Recently, deep autoencoders have gained traction as a powerful method for implementing goal-oriented semantic communications systems. The idea is to train a mapping from the source domain directly to channel symbols, and vice versa. However, prior studies often focused on rate-distortion tradeoff and transmission delay, at the cost of increasing end-to-end complexity and thus latency. Moreover, the datasets used are often not reflective of real-world environments, and the results were not validated against real-world baseline systems, leading to an unfair comparison. In this paper, we study the problem of remote camera pose estimation and propose AdaSem, an adaptive semantic communications approach that optimizes the tradeoff between inference accuracy and end-to-end latency. We develop an adaptive semantic codec model, which encodes the source data into a dynamic number of symbols, based on the latent space distribution and the channel state feedback. We utilize a lightweight model for both transmitter and receiver to ensure comparable complexity to the baseline implemented in a real-world system. Extensive experiments on real-environment data show the effectiveness of our approach. When compared to a real implementation of a client-server camera relocalization service, AdaSem outperforms the baseline by reducing the end-to-end delay and estimation error by over $75\%$ and $63\%$, respectively.
\end{abstract}

\section{Introduction}
The breakthrough in \ac{AI} has promoted the rapid development of intelligent data-driven solutions in many real-world applications, such as immersive experiences, including \ac{AR}, \ac{VR}, and metaverse. These applications are facilitated by a large number of connected mobile and \ac{IoT} devices, generating an enormous amount of data at the network edge. 
To reduce the network traffic load and serve latency-sensitive applications, there is an urgent need to push the \ac{AI} frontiers to the network edge and unleash the potential of big data, i.e., the so-called \emph{edge intelligence}. With ever-increasing computing power, edge devices can either process data locally and send the extracted features to the server for decision-making (enabling hybrid cloud-edge computing solutions), or, directly make local decisions (enabling edge-native solutions). We believe that the former is more pragmatic for a specific type of applications -- those requiring frequent user interactions, but have limited on-device computing or battery capacity, such as \ac{AR}/\ac{VR} devices that share a common metaverse experience.

In parallel, an emerging communications paradigm, called \emph{semantic communications}, has been proposed to alleviate the communication burden, while improving the end-to-end application performance.
The concept was first introduced in Shannon and Weaver's seminal work \cite{shannon1964}, where Weaver identified three levels of problem within the broad subject of communications: 1) the \emph{technical problem} seeks to recover the transmitted symbols as accurately as possible; 2) the \emph{semantic problem}, seeks to recover the meaning of the transmitted symbols; 3) the \emph{effectiveness problem} seeks to affect the conduct of the recipient in the desired way.
In contrast to the \emph{technical} design principles in our current communications systems, semantic and effectiveness communications serve the downstream application directly by extracting and transmitting only the information that is relevant to the common goal of the transmitter and the receiver, leading to both significant reduction in data traffic and improved end-to-end performance. Many works have since described principles and proposed frameworks of a semantic interface between network and application layers \cite{goldreich2012theory,qin2021semantic,luo2022semantic,chaccour2022less}.

Combining edge intelligence with semantic communications, over the past few years, many studies have proposed to solve them jointly through the framework of \ac{JSCC} using techniques found in deep learning. \ac{JSCC} seeks to combine the application layer encoding of information (source coding) with the forward error correction mechanism found in conventional communication systems (channel coding) into a single code, such that the end-to-end application performance is maximized. Using deep learning, \ac{JSCC} can be realized as an autoencoder that maps the source information to channel symbols, and vice versa. Based on simulated results, deep learning-based \ac{JSCC} (DeepJSCC) has shown advantages over the standard communications protocols in terms of both lower distortion and communication overhead in many use cases, such as video transmission \cite{tung2022deepwive}, image retrieval \cite{jankowski2020wireless}, person identification \cite{10066513}, and machine translation \cite{farsad2018deep}. These pioneering works validated DeepJSCC solutions for various tasks, but the model may not generalize to unseen data or dynamic channel conditions. Therefore, many recent works have focused on adaptive solutions, providing architectures that can adapt to time-varying channel conditions and leveraging an information bottleneck to optimize the extraction of semantic features \cite{shao2022learning,dai2022nonlinear,wang2022wireless}. However, these works have mostly achieved good communication rate-distortion tradeoff, at the cost of increasing end-to-end complexity, which ultimately leads to latency and limits their applicability to latency-sensitive applications. Moreover, the baseline solutions and data used to measure the performance of their proposed solutions are often not reflective of real-world conditions, leading to unfair comparisons against the baseline.

In this work, we study the problem of remote camera pose estimation, a.k.a., the camera relocalization problem under the scope of adaptive goal-oriented semantic communications. The goal of camera relocalization is to estimate the 6-\ac{DoF} camera pose, including 3D position and 3D orientation, based on the source information consisting of visual data and motion sensor measurements. Such service is an essential function to realize many immersive visualization-based applications, such as \ac{VR}/\ac{AR} gaming and immersive navigation. 
We propose AdaSem, an adaptive semantic communications approach using deep variational autoencoders. The adaptive codec encodes the source data into a dynamic number of symbols based on channel state feedback and goal-oriented source information, and decodes from the received symbols to the corresponding 6-\ac{DoF} camera pose directly. The main contributions of this paper are summarized as follows:
\begin{itemize}
    \item \emph{End-to-end complexity-awareness}. Unlike the other works that consider only the tradeoff between transmission delay and inference distortion, we aim to develop a light-weight solution that fulfills the end-to-end remote service latency requirement, i.e., including the encoder inference time, transmission delay, and the decoder inference time, while minimizing the inference error. 
    \item \emph{Adaptive codec}. While prior works have assumed a channel model (e.g., \ac{AWGN} channel) with simulated channel conditions, we derive a variational information bottleneck with channel awareness to estimate the communication rate and inference performance tradeoff using real-world channel measurements. We show that our solution is more robust to varying channel conditions in practical testing.
    \item \emph{Realistic performance comparison}.
    To provide a fair comparison to AdaSem, we implemented a real prototype of a client-server camera relocalization service with an Android App and an edge server as the baseline, as shown in Figure~\ref{fig:convComm}. Extensive experiments using real-environment data (including both application and radio data) show that AdaSem outperforms the baseline by reducing the end-to-end delay and estimation error by over $75\%$ and $63\%$, respectively.
\end{itemize}
To the best of the authors' knowledge, this is the first work to develop adaptive semantic communications solutions for remote camera relocalization service, with a performance comparison against a practically implemented baseline, and dataset that reflects real-world usage. 
The rest of the paper is organized as follows. Section \ref{sec:ReWo} goes over the related works. In Section \ref{sec:Problem}, we introduce the system model and formulate the semantic camera relocalization problem. In Section \ref{sec:solution_AdaSem}, we propose the AdaSem solution. The experimental results are shown in Section \ref{sec:Experiment} and conclusions are drawn in Section \ref{sec:Conclusion}.

\begin{figure*}[t]
     \centering
     \begin{subfigure}{.9\textwidth}
         \centering
         \includegraphics[width=0.9\textwidth]{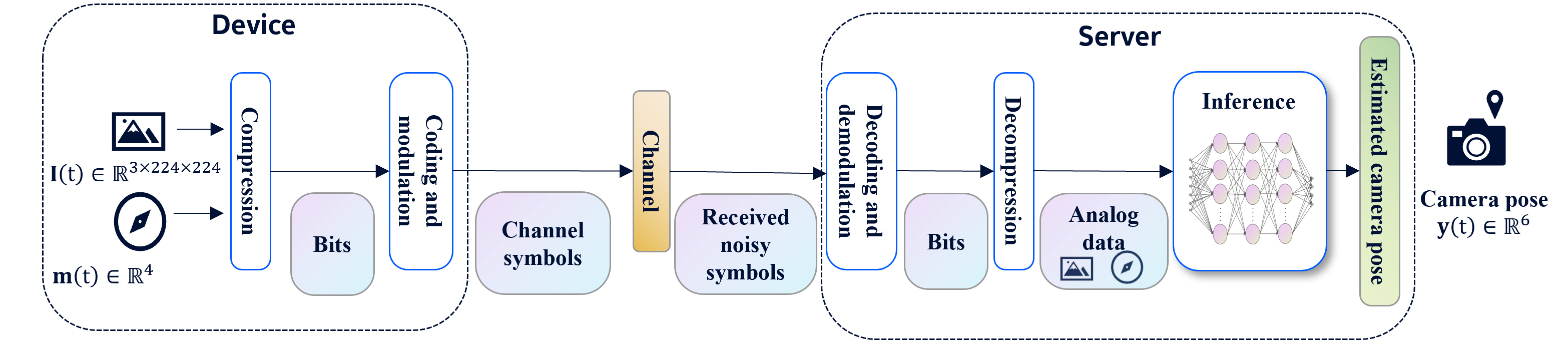}
         \caption{Conventional communications scheme for remote end-to-end camera relocalization.}
         \label{fig:convComm}
     \end{subfigure}
     \begin{subfigure}{.9\textwidth}
         \centering
         \includegraphics[width=0.9\textwidth]{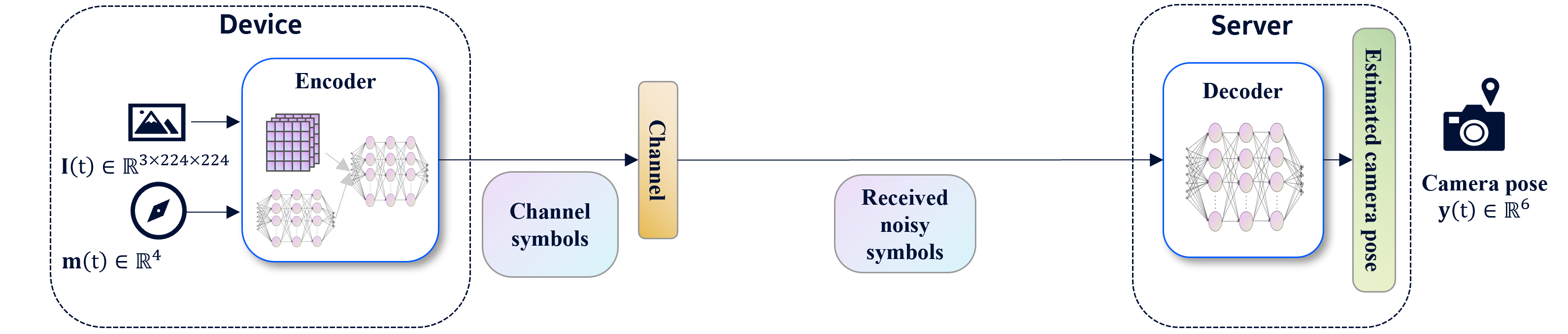}
         \caption{Goal-oriented semantic communications scheme for remote end-to-end camera relocalization.}
         \label{fig:effComm}
     \end{subfigure}
        \caption{Comparison of conventional and semantic communications schemes.}
        \vspace{-3ex}
\end{figure*}

\section{Related Works}\label{sec:ReWo}

{\bf Camera relocalization solutions.}
Conventional camera relocalization techniques include \ac{VO} and \ac{SLAM} \cite{hartley2003multiple}. For example, ORB-SLAM3 \cite{campos2021orb} detects matched features in sequential image frames and estimates the ego-motion within a \ac{RANSAC} scheme. However, their performance is sensitive to fast movement and suffers from slow relocalization when the tracking is lost. Thus, novel camera relocalization methods leverage deep learning techniques, e.g., PoseNet \cite{PoseNet} and its derivatives \cite{GeoPoseNet, LSTMPoseNet,VLocNet++}, have been proposed to learn a direct mapping from the image to the $6$-\ac{DoF} camera pose using \ac{CNN}-based architectures for known scenes. Since a solely visual-based approach is sensitive to camera movement and illumination conditions, side information from motion sensors such as accelerometer and gyroscope, is included in the multi-input \ac{DNN} in \cite{hu2021real}. Besides the above-mentioned end-to-end approaches, hybrid approaches, e.g., PixLoc \cite{sarlin2021back}, have also been proposed where the authors argue that for better generality to new viewpoints or scenes, the \ac{DNN} should focus on learning invariant visual features, while leaving the pose estimation to classical geometric-based algorithms. Note that in the context of remote camera relocalization, the aforementioned works treat the relocalization as a standalone optimization problem, assuming the availability of accurate source data. In this paper, we consider the relocalization service and the communications system 
as a whole, aiming at improving the end-to-end performance in terms of both inference accuracy and end-to-end latency.

{\bf Adaptive semantic communications approaches.} 
In \cite{dai2022nonlinear}, the authors presented \ac{NTSCC}, which utilizes a learned entropy model of the intermediate encoder features to determine the rate of communication needed to achieve a certain distortion, and adapt the transmission rate accordingly. For performance evaluation, the proposed approach was tested on the image transmission task over public datasets such as CIFAR-10 \cite{krizhevsky2009learning}. A similar idea was applied in \cite{wang2022wireless} for the wireless deep video semantic transmission problem, tested over the Vimeo-90K dataset \cite{xue2019video}. A follow-up work \cite{dai2023toward} further extended the adaptive semantic communications concept by introducing a plug-in \ac{CSI} modulation module into the pretrained codec modules, such that the proposed online \ac{NTSCC} can adapt to different source data distributions and channel conditions. In \cite{shao2022learning}, the authors also addressed the adaptive semantic communications problem, but using the deep variational information bottleneck approach inspired by \cite{alemi2017deep}. The performance evaluation was done using the MNIST \cite{deng2012mnist} and CIFAR-10 \cite{krizhevsky2009learning} datasets. However, although these works demonstrated improved transmission bandwidth cost compared to separate source and channel coding approaches, the model complexity can be very high, leading to greater end-to-end latency due to inference cost. 
Moreover, the publicly available datasets provide only the application data, e.g., images or videos and their corresponding labels, without the corresponding radio measurements in the propagation environment. Thus, existing works typically make assumptions regarding the communications system and channel model, which can lead to unrealistic comparisons.


\section{System Model and Problem Statement}\label{sec:Problem}
We consider the remote camera relocalization problem supported by goal-oriented semantic communications. The conventional end-to-end system performs source coding, channel coding, and adaptive modulation as shown in Figure \ref{fig:convComm}, transforming the information across different formats: source data, bits, and symbols.
In contrast, the semantic communications scheme directly encodes the input data into symbols on the device, transmits the encoded symbols, and decodes from the received corrupted symbols into the corresponding camera pose in the remote server, as shown in Figure \ref{fig:effComm}. Table \ref{Tab:notation} summarizes the notations used in this work.

{\bf Source input data.} At each time frame $t\in\N_+$, the following data is collected in the device: 1) image array captured by visual sensor (e.g., camera) $\mI(t)\in\R^{3\times W\times H}$, where $W$ and $H$ are the width and height of the image, respectively; 2) motion sensor measurements $\vm(t)\in\R^n$ extracted from the \ac{IMU} sensors including accelerometer, gyroscope, and magnetometer; 3) received channel feedback $c(t)\in\R$, e.g., \ac{RSRP} measurements. Hereafter we omit the time index $t$ for brevity.

%
{\bf Adaptive encoder.} The adaptive encoder implemented in the device is defined as
\begin{equation}
f_e: \Xs\times\Ks\to \Ss: (\vx, k) \mapsto \vs,
\label{eqn:f_enc}
\end{equation}
where $k$ is the number of symbols given by the adaptation policy, $\Ks$ is a finite set of positive integers, $\vx:=(\mI, \vm)\in\Xs$ is the source data, and $\vs\in\Ss:=\cup_{k\in\Ks}\C^k$ is a vector of complex-valued symbols. We define $\Ss$ as the union of $\C^k, \forall k\in\Ks$, i.e., the encoder allows varying dimensions of $\vs$.

{\bf Adaptation policy.} The number of symbols $k$ is dynamically adapted to the source data $\vx$ and side information, e.g., the channel feedback $c\in\R$. We define the adaptation policy as
\vspace{-1ex}
\begin{equation}
\pi:\Xs\times\R\to\Ks:(\vx, c)\mapsto k.
\label{eqn:pi}
\vspace{-1ex}
\end{equation}
Thus, with \eqref{eqn:f_enc} and \eqref{eqn:pi}, the mapped symbols can be written as $\vs(\vx, c)=f_e(\vx, \pi(\vx, c))$.

{\bf Channel.}  The encoded symbols $\vs\in\Ss$ are transmitted via a noisy channel 
\vspace{-1ex}
\begin{equation}
h:\Ss\to\Ss:\vs\mapsto \hat{\vs},
\label{eqn:channel}
\vspace{-1ex}
\end{equation}
and their noisy versions $\hat{\vs}=h(\vs)$ are received in the server. For example, in \ac{OFDM}, we can assign $k$ \acp{RE} to transmit the $k$ symbols, where each \ac{RE} consists of $1$ subcarrier in the frequency domain and $1$ \ac{OFDM} symbol in the time domain. 

{\bf Adaptive decoder.} An adaptive decoder is implemented in the server to infer the camera pose from the received symbols, given by
\begin{equation}
f_d:\Ss\to \R^6:\hat{\vs}\mapsto\hat{\vy},
\label{eqn:dec}
\end{equation}
where $\hat{\vy}=[\hat{\vp}, \hat{\vo}]$ is the inferred camera pose, composed of the inferred 3D position $\hat{\vp}\in\R^3$ and 3D orientation $\hat{\vo}\in\R^3$. 

\begin{table}[t]
\caption{Table of Notations}
\centering
\label{Tab:notation}
\begin{tabular}{|l|l|}
\hline
Symbol & Meaning \\ \hline
   $\vx$   &     Source data in $\Xs$, composed of image and \ac{IMU} data\\\hline   
   $\vs$   &     Transmitted symbols in $\Ss$  \\\hline
   $\hat{\vs}$   &     Received corrupted symbols in $\Ss$  \\\hline
   $\vy$   &     Camera pose in $\Ys$, composed of position $\vp$ and orientation $\vo$ \\\hline
   $\vp$   &     3D position of the camera \\\hline
   $\vo$   &     3D orientation of the camera \\\hline
   $f_e$   &    Encoder implemented in  device \\\hline
   $\pi$   &    Adaptation policy implemented in device\\\hline
    $k$   &    Dimension of encoded symbols in $\Ks$\\\hline
   $f_d$   &    Decoder implemented in server\\\hline
   $h$   &    Wireless channel \\\hline
\end{tabular}
\end{table}

{\bf Service-oriented goal}. The objective is to minimize the camera relocalization inference distortion $d(\vy, \hat{\vy})$ where $d$ is a defined distance measure, subjected to the \emph{end-to-end latency} constraint  
\begin{equation}
\tau(\pi, f_e, f_d, h)\leq \tau^{\thr}.
\label{eqn:e2e_delay}
\end{equation} 
The end-to-end latency depends jointly on the inference time of the policy, the inference time of the encoder, the transmission time over the channel, and the inference time of the decoder.

Thus, the problem of jointly designing an adaptive encoder $f_e$, the adaptation policy $\pi$, and the adaptive decoder $f_d$, is defined in Problem \ref{prob:EC_Camloc}.
\begin{problem}
\label{prob:EC_Camloc} 
\begin{align}
\minimize_{f_e, f_d, \pi} & ~\Ex_{\vy\sim p(\vy|\vx)}\Ex_{\hat{\vy}\sim p(\hat{\vy}|\vx)}\left[d(\vy, \hat{\vy})\right] \\
 \subject ~ &  
\eqref{eqn:f_enc}, \eqref{eqn:pi}, \eqref{eqn:channel}, \eqref{eqn:dec}, \eqref{eqn:e2e_delay}
\end{align}
\end{problem}

The main challenge of solving Problem \ref{prob:EC_Camloc} is twofold. Firstly, deriving a closed-form expression of Problem \ref{prob:EC_Camloc} with its multiple constraints, is extremely challenging. We, therefore, parameterize the functions $(f_e, f_d)$ using \acp{DNN} and optimize Problem \ref{prob:EC_Camloc} using  a dynamic neural network design \cite{han2021dynamic}. The second challenge is to keep the dynamic architecture lightweight to maintain low end-to-end complexity.

    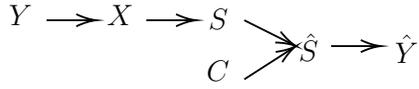
\begin{figure}
        \centering
        \resizebox{0.65\columnwidth}{!}{%
       \tikzset{every picture/.style={line width=0.75pt}} 

\begin{tikzpicture}[x=0.75pt,y=0.75pt,yscale=-1,xscale=1]

\draw    (65.82,34.2) -- (90.82,34.2) ;
\draw [shift={(92.82,34.2)}, rotate = 180] [color={rgb, 255:red, 0; green, 0; blue, 0 }  ][line width=0.75]    (10.93,-3.29) .. controls (6.95,-1.4) and (3.31,-0.3) .. (0,0) .. controls (3.31,0.3) and (6.95,1.4) .. (10.93,3.29)   ;
\draw    (119.82,34.2) -- (144.82,34.2) ;
\draw [shift={(146.82,34.2)}, rotate = 180] [color={rgb, 255:red, 0; green, 0; blue, 0 }  ][line width=0.75]    (10.93,-3.29) .. controls (6.95,-1.4) and (3.31,-0.3) .. (0,0) .. controls (3.31,0.3) and (6.95,1.4) .. (10.93,3.29)   ;
\draw    (172.82,34.2) -- (199.03,47.3) ;
\draw [shift={(200.82,48.2)}, rotate = 206.57] [color={rgb, 255:red, 0; green, 0; blue, 0 }  ][line width=0.75]    (10.93,-3.29) .. controls (6.95,-1.4) and (3.31,-0.3) .. (0,0) .. controls (3.31,0.3) and (6.95,1.4) .. (10.93,3.29)   ;
\draw    (219.82,48.2) -- (244.82,48.2) ;
\draw [shift={(246.82,48.2)}, rotate = 180] [color={rgb, 255:red, 0; green, 0; blue, 0 }  ][line width=0.75]    (10.93,-3.29) .. controls (6.95,-1.4) and (3.31,-0.3) .. (0,0) .. controls (3.31,0.3) and (6.95,1.4) .. (10.93,3.29)   ;
\draw    (172.82,66.2) -- (199.14,49.28) ;
\draw [shift={(200.82,48.2)}, rotate = 147.26] [color={rgb, 255:red, 0; green, 0; blue, 0 }  ][line width=0.75]    (10.93,-3.29) .. controls (6.95,-1.4) and (3.31,-0.3) .. (0,0) .. controls (3.31,0.3) and (6.95,1.4) .. (10.93,3.29)   ;

\draw (44,24.4) node [anchor=north west][inner sep=0.75pt]  [font=\large]  {$Y$};
\draw (97,24.4) node [anchor=north west][inner sep=0.75pt]  [font=\large]  {$X$};
\draw (152,26.4) node [anchor=north west][inner sep=0.75pt]  [font=\large]  {$S$};
\draw (200.82,39.6) node [anchor=north west][inner sep=0.75pt]  [font=\large]  {$\hat{S}$};
\draw (253,40.4) node [anchor=north west][inner sep=0.75pt]  [font=\large]  {$\hat{Y}$};
\draw (151,54.4) node [anchor=north west][inner sep=0.75pt]  [font=\large]  {$C$};

\end{tikzpicture}
        }
				\caption{Assumed Markov chain of Problem \ref{prob:EC_Camloc}.}
    \label{fig:MarkovAssumption}
    \end{figure}
    
\section{Proposed AdaSem Solution}\label{sec:solution_AdaSem}
To solve Problem \ref{prob:EC_Camloc}, we propose AdaSem, which adaptively optimizes the number of the transmitted symbols for each source sample based on the channel-aware \ac{VIB}. 
The challenge is to find a good tradeoff between the communication rate and the distortion, by exploiting the channel feedback, while maintaining a low-complexity encoder and decoder. The proposed AdaSem solution consists of the following three components: channel-aware \ac{VIB}, adaptation policy, and a lightweight dynamic neural network architecture.

\subsection{Channel-Aware variational information bottleneck}
In \cite{alemi2017deep}, the deep variational information bottleneck method was developed as the variational version of the information bottleneck of Tishby \cite{tishby2000information}. We use a similar principle for our channel-aware semantic communications problem, where the bottleneck helps us make the tradeoff between the communication rate and the inference distortion.

We assume a Markov chain as shown in Figure \ref{fig:MarkovAssumption}, where $Y, X, S, \hat{S}, \hat{Y}, C$ denote the random variables of the inference target $\vy$, input data $\vx$, encoded symbol $\vs$, corrupted symbol $\hat{\vs}$, inferred output $\hat{\vy}$, and channel feedback $c$, respectively. Namely, we assume $p(\hat{S}|X, Y, C) = p(\hat{S}|X, C)$, i.e., the received symbols depends directly on $X$ and $C$ only, but not on $Y$. 

With the assumed Markov chain, the information bottleneck can be written as 
\begin{equation}
\Ls^{\ib} := \underbrace{-I(\hat{S}; Y)}_{\mbox{Distortion}} + \beta \underbrace{I(\hat{S}; X, C)}_{\mbox{Rate}},
\label{eqn:IB}
\end{equation}
where the first term is the negative mutual information between the received symbols and the inference target, seen as the \emph{distortion}, while the second term is the mutual information between the received symbols and the observations of both source data and channel feedback, seen as the \emph{rate}.
    
Problem \ref{prob:EC_Camloc} is related to minimizing \eqref{eqn:IB}, because the objective is also to minimize the distortion, with an appropriate $\beta$ such that the number of transmitted symbols $k$ (reflecting the communication rate) leads to a feasible end-to-end delay, as both the inference time and transmission time are monotonically non-decreasing functions of $k$. Thus, following the similar method as in \cite{alemi2017deep}, we derive the following corollary for the approximation of Equation \eqref{eqn:IB}.
    
\begin{corollary}
Assume the Markov chain in Figure \ref{fig:MarkovAssumption}, let $q_{\phi}(\vy|\hat{\vs})$ (the decoder) be a variational approximation to  $p(\vy|\hat{\vs})$. Let $p_{\theta}(\hat{\vs}|\vx,c)$ (the encoder) be a variational approximation to $p(\hat{\vs}|\vx, c)$, and let $r(\hat{\vs})$ be a variational approximation to the marginal $p(\hat{\vs})=\int_{\vx,c} p(\hat{\vs}|\vx, c)p(\vx)p(c)d\vx \, dc $, an upper bound of \eqref{eqn:IB} is given by
\begin{align}
\Ls^{\vib}:= & \Ex_{p(\vx, \vy, c)}\Big(-\Ex_{p_{\theta}(\hat{\vs}|\vx, c)}\left[\log q_{\phi}(\vy|\hat{\vs})\right] \nonumber \\
 & + \beta \KL\big(p_{\theta}(\hat{\vs}|\vx, c)||r(\hat{\vs})\big)\Big) \geq \Ls^{\ib}  
 \label{eqn:vib_upperbound}
\end{align}
where $\KL(p||q)$ denote the \ac{KL} divergence between the two distributions $p$ and $q$.
\label{corol:VIB_compute}
\end{corollary}

\begin{proof}
The proof is given in Appendix \ref{sec:coro_proof}.
\end{proof}

Note that $p(\vx, \vy, c)$ can be approximated using the empirical data distribution $p(\vx, \vy, c)=(1/N) \sum_{i=1}^N \delta_{\vx_i}\delta_{\vy_i}\delta_{c_i}$, where $\delta$ is a delta distribution. We can then approximate $\Ls^{\vib}$ as
\begin{align}
\Ls^{\vib}  \approx & \frac{1}{N}  \sum_{i=1}^N \Big(-\Ex_{\boldsymbol{\epsilon}\sim p(\boldsymbol{\epsilon})} \big[\log q_{\phi}(\vy_i|g(\vx_i, c_i, \boldsymbol{\epsilon}))\big] \nonumber \\
& + \beta \KL(p_{\theta}(\hat{\vs}|\vx_i, c_i)|| r(\hat{\vs}))\Big),
\label{eqn:empirical_vib}
\end{align}
and we use a multivariate Gaussian distribution $\boldsymbol{\epsilon}$ to reparameterize the marginal distribution $p(\hat{\vs})$, such that $p_{\theta}(\hat{\vs}|\vx, c) = \Ns\left(\hat{\vs}|g^{\vmu}(\vx, c), g^{\msig}(\vx, c)\right)$. Thus, we simply need to find a deterministic function $g(\vx_i, c_i)$ over $(X, C)$ that predicts the means $g^{\vmu}(\vx, c)$ and covariance $g^{\msig}(\vx, c)$ of the multivariate Gaussian distribution.
    
\subsection{AdaSem model}
With \eqref{eqn:empirical_vib}, the \ac{VIB} in \eqref{eqn:vib_upperbound} is tractable and can be computed when training an autoencoder model. Therefore, we propose a model as shown in Figure \ref{fig:AdaSem_Autoencoder}, which has the following four components:
\begin{itemize}
\item {\bf Variational encoder} $g = f_{\hat{s}}\circ f_z$ estimates the distribution of the corrupted symbols and provides the communication rate estimation to the adaptation policy for optimizing the symbol dimension $k$.
\item {\bf Adaptation policy} $\pi$ chooses the optimal $k^\ast \in \mathcal{K}$ and provides it to the encoder. 
\item {\bf Encoder} $f_e=f_s\circ f_z$ encodes $(\vx, c)$ into varying lengths of symbols $\vs$ to transmit, based on the optimized $k$.
\item {\bf Decoder} $f_d$ receives noisy symbols $\hat{\vs}$ and infer the camera pose $\vy$ in server.
\end{itemize}
Note that $f_e, g, \pi$ are implemented in the device, while $f_d$ is implemented in the server. Thus, it is critical to keep the models in the device lightweight. To this end, we propose to utilize a known lightweight pretrained model MobileNetV3 \cite{howard2019searching} for the image feature extractor. The basic architecture of $f_e$ is shown in Figure \ref{fig:BasicEncoder}. The image data is fed into a \ac{CNN}, initialized using the lower layer weights of a pretrained MobileNetV3 \cite{howard2019searching}. The reason is the following: first, the useful features of an image for inferring geometric relations are usually the keypoints, such as edges and corners \cite{campos2021orb}, which are assumed to be extracted by the lower layers of an object detector. Second, by only using a few layers from a pretrained object detector, we can reduce the complexity of the encoder and thus the inference time. The \ac{IMU} data is processed with a \ac{MLP} consisting of fully connected layers. The above-mentioned two branches are concatenated and fed into $f_s$ and $f_{\hat{s}}$ for channel input encoding and communication rate adaptation, respectively. 
Both $f_{\hat{s}}$ and $f_s$ utilize lightweight architectures, made of one fully connected layer only.
The details of the adaptation method will be discussed in the following subsection.

\begin{figure}
        \centering
        \resizebox{.9\columnwidth}{!}{%
       \input{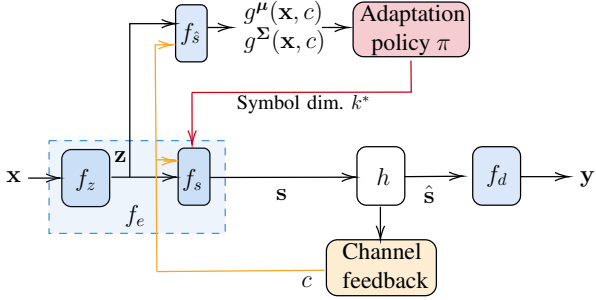}
        }
				\caption{Proposed AdaSem model consisting of 1) encoder $f_e=f_s\circ f_z$, 2) variational encoder $g = f_{\hat{s}}\circ f_z$, 3) adaptation policy $\pi$, and 4) decoder $f_d$.}
    \label{fig:AdaSem_Autoencoder}
    \end{figure}




For a better-customized service experience, we define an application-oriented distortion measure $d^{\app}(\vy, \hat{\vy})$ as a weighted sum of the position error and the angular distance.    
\begin{equation}
d^{\app}(\vy, \hat{\vy}) := (1-\alpha)\|\vp-\hat{\vp}\|_2 - \alpha |\vq\cdot \hat{\vq}'|,
\label{eqn:pose_loss}
\end{equation}
where $\hat{\vp}$ is the inferred position and $\hat{\vq}\in\R^4$ is the \emph{quaternion}, which is an equivalent form of the inferred orientation $\hat{\vo}$. Because the quaternion is subjected to unit length, we denote $\hat{\vq}':= \hat{\vq}/\|\hat{\vq}\|$ as the normalized quaternion.
By tuning the weight $\alpha$, we can prioritize either the position or the orientation accuracy, depending on the individual application's requirement.
Note that instead of the Euclidean distance used in PoseNet \cite{PoseNet}, we leverage \emph{quaternion} to better represent the angular distance. In Appendix \ref{sec:agular} we show how the efficient computation $|\vq\cdot \hat{\vq}'|$ reflects the angular distance. 

At the training phase, $f_z, f_s, f_{\hat{s}}, f_d$ are jointly trained, using the end-to-end loss function
\begin{equation}
\Ls = d^{\app}(\vy, \hat{\vy}) + \eta \, \Ls^{\vib}, 
\label{eqn:loss}
\end{equation}
where $\Ls^{\vib}$ and $d^{\app}(\vy, \hat{\vy})$ are given in Equations \eqref{eqn:vib_upperbound} and \eqref{eqn:pose_loss} respectively, and $\eta$ is the tuning factor between the \ac{VIB} and the application-oriented distortion.

\subsection{Adaptation policy}\label{ssec:policy}
To minimize the computational cost and delay in the device, we propose a simple heuristic policy to adjust the number of symbols $k$ based on the derived upper bound of the communication rate in Equation \eqref{eqn:IB}.
\begin{figure}[t]
     \centering
         \includegraphics[width=0.49\textwidth]{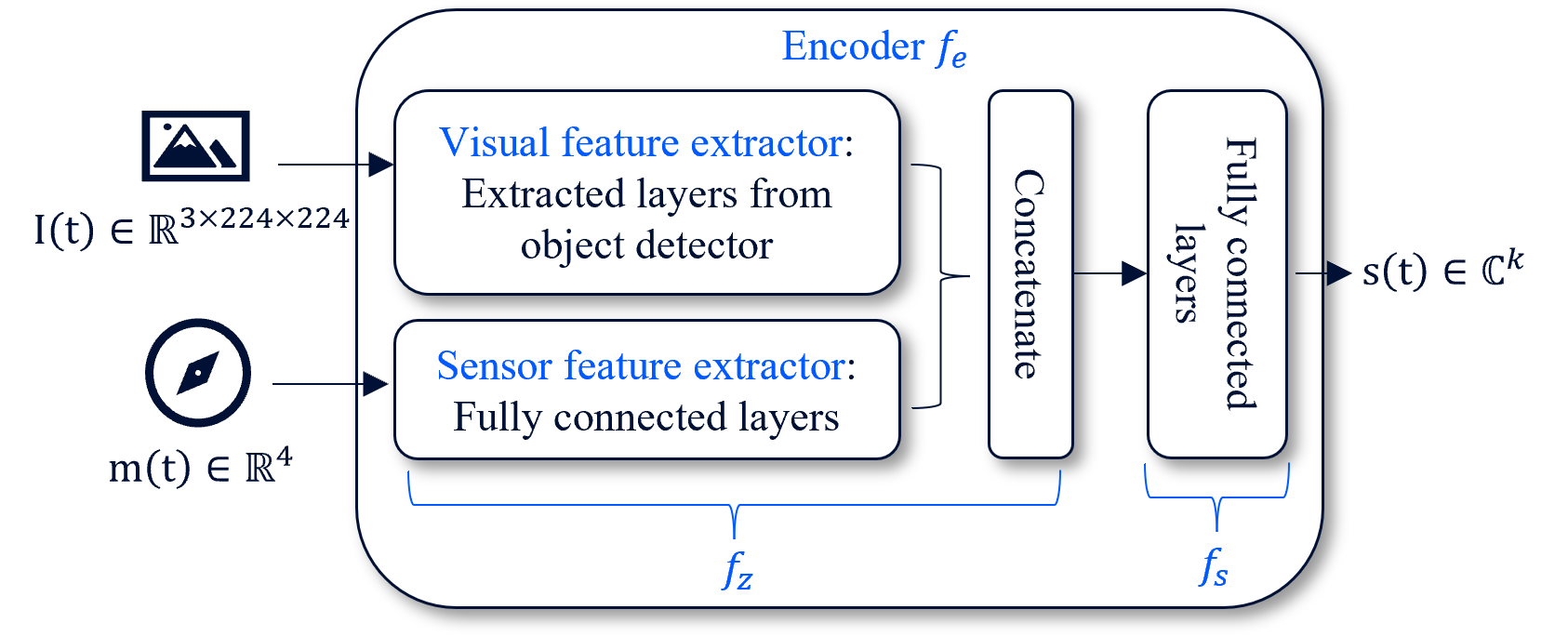}
         \caption{Basic architecture of the multi-input encoder $f_e$.}
\label{fig:BasicEncoder}
\end{figure}
This is because, the communication rate is directly related to the mutual information between the source data and side information $(X, C)$ and the estimated corrupted symbols $\hat{S}$. Thus, we define the estimated $\bar{k}$ proportional to the communication rate upper bound: 
\begin{equation}
\bar{k}(\vx_i, c_i) = \gamma  \KL(p_{\theta}(\hat{\vs}|\vx_i, c_i)|| r(\hat{\vs})), \forall i,
\label{eqn:k_prop}
\end{equation}
where $\gamma$ is a tuning factor.

Moreover, as we will specify in the following subsection, we design a dynamic architecture such that the model complexity is non-decreasing over the symbol dimension $k \in \mathcal{K}$. We also assume that given the same channel state and radio resource, the transmission latency is monotonically non-decreasing over $k$. Thus, the end-to-end latency is a monotonically non-decreasing function of $k$, denoted by $\tau(k)$. With the given architecture, the inference time of the autoencoder can be estimated in advance. Also, with the channel feedback and given radio resource, the transmission delay can be also estimated. Thus, we assume that $\tau(k)$ is a known monotonically non-decreasing function over $k$. The policy is then simply to find
\begin{equation}
k^{\ast}(\vx_i, c_i) = \argmax_{k\in\Ks} \left\{k: k\leq \bar{k}(\vx_i, c_i); \tau(k)\leq \tau^{\thr}\right\}.
\end{equation}
This is easy to find using the non-decreasing property of $\tau(k)$, i.e., if $\tau(\bar{k})\geq \tau^{\thr}$, then stepwise reduce $\bar{k}$ until $\tau(\bar{k})\leq \tau^{\thr}$, otherwise $k^{\ast}=\bar{k}$. 

\subsection{Dynamic neural network architecture}\label{ssec:dynamicNN}
Although there are many options for dynamic neural network architectures \cite{han2021dynamic}, our challenge is to achieve a low end-to-end latency, which means developing a low-complexity architecture. This makes the rate adaptation networks in \cite{dai2022nonlinear,wang2022wireless,shao2022learning} not suitable, as the inference time for any $k\in\Ks$ is similar, due to the constant network size. 

We take a practical assumption that the set of possible numbers of symbols $\Ks$ is a small discrete finite set. This assumption is reasonable, as in typical wireless communication systems, the number of modulation schemes is usually small as well. Thus, we design a multi-head architecture for the symbol encoding layers $f_s$, as shown in Figure \ref{fig:adalayer_enc}. Given $\Ks:=\{k_0, \ldots, k_M\}$, if a particular dimension $k^{\ast}$ is chosen, say, if $k^{\ast}=k_1$, then, only the output head with $2k_1$ output dimension is activated. This means that as $k^\ast$ reduces, the number of \ac{FLOP}s needed to compute the channel input symbols is also reduced. Moreover, the multi-head structure also improves training stability, as we discovered empirically.
Note that, to generate the $k$ complex symbols, we generate a $2k$-dimensional real-valued vector, such that consecutive values are paired to form the real and imaginary parts of a complex symbol.
Moreover, to ensure that each symbol in $\vs$ is subjected to a maximum unit transmission power, the output layer has a Tanh activation function and each symbol $\vs[j], j= 1, \ldots, k$ in $\vs$ is scaled with $\min\{1, 1/\|\vs[j]\|\}$.

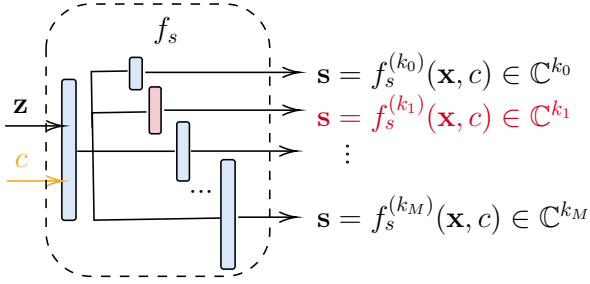
\begin{figure}
        \centering
        \resizebox{.9\columnwidth}{!}{%
       \tikzset{every picture/.style={line width=0.75pt}} 

\begin{tikzpicture}[x=0.75pt,y=0.75pt,yscale=-1,xscale=1]

\draw  [fill={rgb, 255:red, 74; green, 144; blue, 226 }  ,fill opacity=0.2 ] (79.44,116.58) .. controls (79.44,115.58) and (80.25,114.78) .. (81.24,114.78) -- (86.64,114.78) .. controls (87.63,114.78) and (88.44,115.58) .. (88.44,116.58) -- (88.44,202.51) .. controls (88.44,203.5) and (87.63,204.31) .. (86.64,204.31) -- (81.24,204.31) .. controls (80.25,204.31) and (79.44,203.5) .. (79.44,202.51) -- cycle ;
\draw  [fill={rgb, 255:red, 74; green, 144; blue, 226 }  ,fill opacity=0.2 ] (122.44,102.38) .. controls (122.44,101.49) and (123.16,100.78) .. (124.04,100.78) -- (128.84,100.78) .. controls (129.72,100.78) and (130.44,101.49) .. (130.44,102.38) -- (130.44,119.71) .. controls (130.44,120.59) and (129.72,121.31) .. (128.84,121.31) -- (124.04,121.31) .. controls (123.16,121.31) and (122.44,120.59) .. (122.44,119.71) -- cycle ;
\draw  [fill={rgb, 255:red, 208; green, 2; blue, 27 }  ,fill opacity=0.2 ] (134.94,121.04) .. controls (134.94,120.21) and (135.61,119.54) .. (136.44,119.54) -- (140.94,119.54) .. controls (141.77,119.54) and (142.44,120.21) .. (142.44,121.04) -- (142.44,147.81) .. controls (142.44,148.64) and (141.77,149.31) .. (140.94,149.31) -- (136.44,149.31) .. controls (135.61,149.31) and (134.94,148.64) .. (134.94,147.81) -- cycle ;
\draw  [fill={rgb, 255:red, 74; green, 144; blue, 226 }  ,fill opacity=0.2 ] (180.44,167.58) .. controls (180.44,166.58) and (181.25,165.78) .. (182.24,165.78) -- (187.64,165.78) .. controls (188.63,165.78) and (189.44,166.58) .. (189.44,167.58) -- (189.44,233.51) .. controls (189.44,234.5) and (188.63,235.31) .. (187.64,235.31) -- (182.24,235.31) .. controls (181.25,235.31) and (180.44,234.5) .. (180.44,233.51) -- cycle ;
\draw    (132.44,110.31) -- (228.44,110.31) ;
\draw [shift={(230.44,110.31)}, rotate = 180] [color={rgb, 255:red, 0; green, 0; blue, 0 }  ][line width=0.75]    (10.93,-3.29) .. controls (6.95,-1.4) and (3.31,-0.3) .. (0,0) .. controls (3.31,0.3) and (6.95,1.4) .. (10.93,3.29)   ;
\draw  [fill={rgb, 255:red, 74; green, 144; blue, 226 }  ,fill opacity=0.2 ] (152.44,143.38) .. controls (152.44,142.49) and (153.16,141.78) .. (154.04,141.78) -- (158.84,141.78) .. controls (159.72,141.78) and (160.44,142.49) .. (160.44,143.38) -- (160.44,177.71) .. controls (160.44,178.59) and (159.72,179.31) .. (158.84,179.31) -- (154.04,179.31) .. controls (153.16,179.31) and (152.44,178.59) .. (152.44,177.71) -- cycle ;
\draw    (144.44,134.31) -- (228.44,134.31) ;
\draw [shift={(230.44,134.31)}, rotate = 180] [color={rgb, 255:red, 0; green, 0; blue, 0 }  ][line width=0.75]    (10.93,-3.29) .. controls (6.95,-1.4) and (3.31,-0.3) .. (0,0) .. controls (3.31,0.3) and (6.95,1.4) .. (10.93,3.29)   ;
\draw    (161.44,160.31) -- (226.44,160.31) ;
\draw [shift={(228.44,160.31)}, rotate = 180] [color={rgb, 255:red, 0; green, 0; blue, 0 }  ][line width=0.75]    (10.93,-3.29) .. controls (6.95,-1.4) and (3.31,-0.3) .. (0,0) .. controls (3.31,0.3) and (6.95,1.4) .. (10.93,3.29)   ;
\draw    (191.44,202.31) -- (226.44,202.31) ;
\draw [shift={(228.44,202.31)}, rotate = 180] [color={rgb, 255:red, 0; green, 0; blue, 0 }  ][line width=0.75]    (10.93,-3.29) .. controls (6.95,-1.4) and (3.31,-0.3) .. (0,0) .. controls (3.31,0.3) and (6.95,1.4) .. (10.93,3.29)   ;
\draw    (89.16,160.43) -- (151.16,160.43) ;
\draw    (43.44,144.31) -- (78.44,144.31) ;
\draw [shift={(80.44,144.31)}, rotate = 180] [color={rgb, 255:red, 0; green, 0; blue, 0 }  ][line width=0.75]    (10.93,-3.29) .. controls (6.95,-1.4) and (3.31,-0.3) .. (0,0) .. controls (3.31,0.3) and (6.95,1.4) .. (10.93,3.29)   ;
\draw [color={rgb, 255:red, 245; green, 166; blue, 35 }  ,draw opacity=1 ]   (44.44,179.31) -- (78.44,179.31) ;
\draw [shift={(80.44,179.31)}, rotate = 180] [color={rgb, 255:red, 245; green, 166; blue, 35 }  ,draw opacity=1 ][line width=0.75]    (10.93,-3.29) .. controls (6.95,-1.4) and (3.31,-0.3) .. (0,0) .. controls (3.31,0.3) and (6.95,1.4) .. (10.93,3.29)   ;
\draw   (98.73,145.44) -- (98.15,109.81) -- (121.16,109.43) ;
\draw   (99.16,160.43) -- (98.76,136) -- (133.16,135.43) ;
\draw   (180.3,202.87) -- (99.88,204.2) -- (99.16,160.43) ;
\draw  [dash pattern={on 4.5pt off 4.5pt}] (69.44,95.29) .. controls (69.44,79.6) and (82.16,66.89) .. (97.84,66.89) -- (183.04,66.89) .. controls (198.72,66.89) and (211.44,79.6) .. (211.44,95.29) -- (211.44,211.49) .. controls (211.44,227.17) and (198.72,239.89) .. (183.04,239.89) -- (97.84,239.89) .. controls (82.16,239.89) and (69.44,227.17) .. (69.44,211.49) -- cycle ;

\draw (136,73.81) node [anchor=north west][inner sep=0.75pt, font=\Large]    {$f_s$};
\draw (47,125.81) node [anchor=north west][inner sep=0.75pt, font=\Large]    {$\vz$};
\draw (48,160.81) node [anchor=north west][inner sep=0.75pt, font=\Large]  [color={rgb, 255:red, 245; green, 166; blue, 35 }  ,opacity=1 ]  {$c$};
\draw (158.52,183.66) node [anchor=north west][inner sep=0.75pt, font=\Large]  [rotate=-358.66]  {$...$};
\draw (240,95.81) node [anchor=north west][inner sep=0.75pt, font=\Large]    {$\vs=f_s^{( k_{0})}(\vx, c)\in\C^{k_0}$};
\draw (240,186.81) node [anchor=north west][inner sep=0.75pt, font=\Large]    {$\vs=f_s^{( k_{M})}(\vx, c)\in\C^{k_M}$};
\draw (240,122.81) node [anchor=north west][inner sep=0.75pt, font=\Large]  [color={rgb, 255:red, 208; green, 2; blue, 27 }  ,opacity=1 ]  {$\vs=f_s^{( k_{1})}(\vx, c)\in\C^{k_1}$};
\draw (263.32,153.07) node [anchor=north west][inner sep=0.75pt, font=\Large]  [rotate=-90.15]  {$...$};

\end{tikzpicture}
        }
				\caption{Multi-head encoder. E.g., if $k^{\ast}=k_1$, then, only the head corresponding to $k_1$ is activated.}
    \label{fig:adalayer_enc}
    \end{figure}

As for the decoder $f_d$,  we use an \ac{MLP}, to map the received noisy symbols $\hat{\vs}$ to the predicted pose $\vy$ directly. 
The input layer is $2k_M$-dimensional (including both real and imaginary parts), where $k_M$ is the maximum number of symbols. If a decoder receives $k$ symbols, it simply zero-pads the symbols to a $2k_M$-dimensional vector and uses it as the input. Note that, unlike the multi-head encoder, in the decoder we do not use the multi-input design. Therefore, the decoder inference time is almost the same for all $k\in\Ks$. The reason is that the decoder has a simpler \ac{MLP} architecture with very low inference time, and it is implemented in the server with more powerful computational capability. Moreover, experimental results show that using the zero-padding method in the decoder results in faster convergence than the multi-input architecture.
Overall, the proposed AdaSem algorithms for the training and inference phase are summarized in Algorithm \ref{algo:AdaSem} and \ref{algo:AdaSem_inf}, respectively.

\begin{algorithm}[h]
    \caption{AdaSem model training}
    \label{algo:AdaSem}
      \textbf{Input} $T$ number of epochs, $B$ batch size, dataset $\Ds:=\{\vx_i, c_i, \vy_i\}_{i=1}^N$, tuning factors $\alpha, \beta, \gamma, \eta$, delay threshold $\tau^{\thr}$\\
      \vspace{-3ex}
    \begin{algorithmic}[1]  
    \For{epoch $t=1, ..., T$}
    \For{minibatch $b = 1, ..., \lfloor N/B \rfloor$}
        \For{sample $i = 1, \ldots, B$}
        \State Choose $k_i$ based on \eqref{eqn:k_prop}
        \EndFor
        \State Compute application loss $d^{\app}$ based on \eqref{eqn:pose_loss} 
        \State Compute per minibatch $\Ls^{\vib}$ based on \eqref{eqn:empirical_vib} by
        \State replacing $N$ by $B$
        \State Compute loss based on \eqref{eqn:loss}
        \State Update weights and biases of $f_z, f_s, f_{\hat{s}}, f_d$
    \EndFor
    \EndFor
    \end{algorithmic}
\end{algorithm}
\begin{algorithm}[h]
    \caption{AdaSem inference}
    \label{algo:AdaSem_inf}
      \textbf{Input} Test inputs $\{\vx_i, c_i\}_{i=1}^{N'}$, tuning factors $\alpha, \beta, \gamma, \eta$, delay threshold $\tau^{\thr}$\\
      \vspace{-3ex}
    \begin{algorithmic}[1]  
    \For{sample $i = 1, \ldots, N'$}
        \State Collect source data $\vx_i$ and channel feedback $c_i$ 
        \State Choose $k_i$ based on \eqref{eqn:k_prop}
        \State Compute symbols $\vs_i\leftarrow f_e(\vx_i, c_i)$
        \State Transmit symbols $\vs_i$ over channel $h(\cdot)$
        \State Receive noisy symbols $\hat{\vs}_i \leftarrow h(\vs)$
        \State Decode and infer the camera pose $\hat{\vy}\leftarrow f_d(\hat{\vs})$        
    \EndFor
    \end{algorithmic}
\end{algorithm}
\section{Experimental Results}\label{sec:Experiment}
To compare the performance in a realistic scenario, we evaluate the performance of AdaSem against a real implemented baseline -- a remote camera relocalization service for \ac{AR}-supported radio map visualization application \cite{liao2022haru}.
In this section, we first describe the baseline system setup and data collection. Then, we provide further details on the implemented neural network architecture and hyperparameters for the proposed AdaSem algorithm. Finally, the performance comparison, feature interpretation, and evaluation of the adaptation policy will be presented.  
\subsection{Baseline system setup and data collection}
We consider a remote camera relocalization service for \ac{AR}-supported radio map visualization application. The baseline system is implemented with a self-developed Android App on a mobile device and a neural network for camera pose inference in the edge server as shown in Figure \ref{fig:e2e_app}. We use a Nokia $6$ as the mobile device, and Levono Thinkpad P43s equipped with a Quadro-P520 GPU as the edge server for real-time inference. However, the neural network for camera pose inference is pretrained on a cloud server equipped with $8$ Nvidia Tesla K80 GPUs using the data collected by the App. 

The baseline system follows the standard communications protocols as shown in Figure \ref{fig:convComm}. The communication between the device and the server is via a dedicated Wi-Fi integrated with a Nokia FW2HC omni antenna, with an antenna gain of 4.7 dBi, operated at 2.4 GHz. The App in the device captures the RGB camera stream with a resolution of $480$p and the motion sensor measurements of 4D relative quaternion derived from Android attitude composite sensors \cite{AndroidSensors}.  
The codec for the camera stream is H.265 and a low-pass filter is applied to the sensor output to reduce the noise artifacts and smooth the signal reading. This is because mobile devices usually use low-cost \ac{IMU} sensors. We use a custom extension to \ac{RTP} to include both the compressed camera stream and sensor data in the \ac{RTP} packets to transmit over the wireless channel. The received data is channel and source decoded, and used as the input to a \ac{DNN} to infer the camera pose. Note that for a fair comparison, the \ac{DNN} follows the same structure as the semantic encoder in Figure \ref{fig:BasicEncoder}, except that the output is the camera pose instead of the transmission symbols. The inferred camera pose is sent back to the device for the rendering of the augmented radio map. Note that the \ac{AR} service is outside of the scope of this paper.  

\begin{figure}[t]
     \centering
\includegraphics[width=0.48\textwidth]{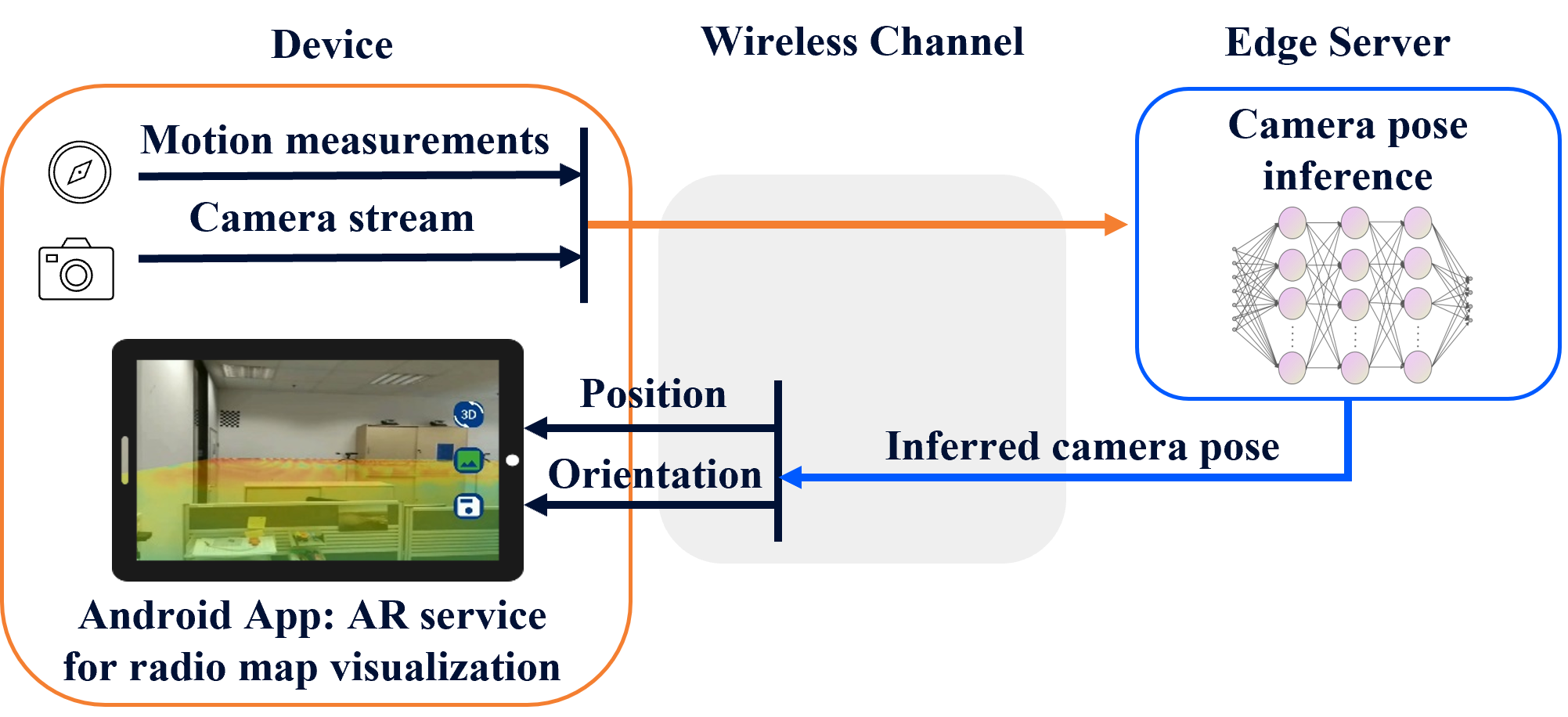}
         \caption{Implemented baseline system (non-semantic).}
\label{fig:e2e_app}
\end{figure}
We conducted the experiment in a real office environment of a size of $5.8 \mbox{m}\times 5.2 \mbox{m} \times 3.85 \mbox{m}$. The developed App can collect the camera and motion sensor data at a rate of $30$ fps and read the signal strength every $30$ ms. The remaining problem is how to derive the ground truth of the camera pose. We first leverage the visual-inertial ORB-SLAM3 \cite{campos2021orb} algorithm to automatically generate the relative camera poses $\vy(t)$ with respect to a random reference coordination system. Then, we transform them into the absolute world coordinate system by utilizing markers placed in the room to accurately identify the room edges. In total, we collected $67681$ data samples. 
To generalize the data for arbitrary viewpoints, we generated a 3D model as a digital twin of the room and computed the radio map for the whole environment using ray-tracing \cite{hoppe2017wave}. The ray-tracing generated signal strength and an assumed AWGN channel are then used for training, validation, and testing instead of the spatially sparse real radio measurements, with a sample split ratio of $60/20/20$. That is, we assume $\hat{\vs}(t) = \vs(t) + \mathbf{n}(t)$, where $\mathbf{n}(t) \sim CN(0, \sigma_n^2 (t) \mathbf{I})$ is an independent and identically distributed (i.i.d.) complex Gaussian random vector with  $\sigma_n^2(t) = 1 / \text{SNR}(t)$, where $\text{SNR}(t)$ is the channel \ac{SNR} at time $t$, as generated by the ray-tracing algorithm using the measured signal strength $c(t)$. Note that this assumption is different from those in prior works, e.g., \cite{dai2022nonlinear}, where the channel model and condition are purely simulated. In our work, we assume an AWGN channel model but utilize measured \ac{SNR} at various points in the room for training and testing, which is not only more realistic but also matches the assumptions that the Wi-Fi system in the baseline makes.
%
%
\subsection{Neural network architecture and training hyperparameters}
The designed architecture of AdaSem is provided in Table \ref{Tab:architecture}, where we use $(n_0, n_1, \ldots, n_D)$ to denote $D$ fully connected layers with $n_d$ as the number of neurons in a layer. The input layer $n_0$ is not counted in $D$. Except for the output layers, all the fully connected layers use the Leaky ReLU activation before feeding into the next layer. We assume a finite set of the optional numbers of symbols $\Ks:=\{64, 128, 192, 256, 320, 384, 448, 512\}$, i.e., $M = |\Ks| = 8$ and the maximum number is $K_M = 512$. Recall that the output layers of encoder $f_s$ is $2k$ if $k$ is the value chosen by the adaptation policy, because it includes both the real and imaginary parts. The output layer of the variational encoder $f_{\hat{s}}$ is $4 k_M$ because it outputs both $g^{\vmu}(\vx, c)$ and $\ve{\sigma}(\vx, c)$, where we assume $g^{\msig}(\vx, c)$ is a diagonal matrix with diagonal $\ve{\sigma}(\vx, c)$. 

We use $\alpha=0.7$, $\beta=1$, $\eta=0.2$, and $\gamma=100$, for these hyperparameters shown in Section \ref{sec:solution_AdaSem}. The end-to-end constraint is defined by $\tau^{\thr}=16$ ms per request as we target an \ac{AR} service of $60$ fps. For training, we use the Adam optimizer with a learning rate of $0.0001$ and batch size of $64$. 

%
\begin{table}[t]
\caption{Architecture of AdaSem}
\centering
\label{Tab:architecture}
\begin{tabular}{lllllll}
\cline{1-3}
\multicolumn{1}{|l|}{} & \multicolumn{1}{l|}{Layers}  & \multicolumn{1}{l|}{\begin{tabular}[c]{@{}l@{}}Output \\ Dimensions\end{tabular}} &  &  &  &  \\ \cline{1-3}
\multicolumn{1}{|l|}{\begin{tabular}[c]{@{}l@{}}AdaSem  \\ Encoder\end{tabular}} & \multicolumn{1}{l|}{\begin{tabular}[c]{@{}l@{}}Visual feature extractor: \\ the lowest 3 layers of MobileNetV3\end{tabular}}  & \multicolumn{1}{l|}{$28\times 28 \times 24 $}                                     &  &  &  &  \\ \cline{2-3}
\multicolumn{1}{|l|}{}  & \multicolumn{1}{l|}{\begin{tabular}[c]{@{}l@{}}Sensor feature extractor: \\ two fully connected layers $(4, 8, 4)$\end{tabular}} & \multicolumn{1}{l|}{4}        &  &  &  &  \\ \cline{2-3}
\multicolumn{1}{|l|}{}   & \multicolumn{1}{l|}{\begin{tabular}[c]{@{}l@{}}$f_z$: Concatenator added by two \\  fully connected layers \\ $(28\times28\times 24 + 4, 4k_M, 2k_M)$\end{tabular}} & \multicolumn{1}{l|}{$2k_M$}  &  &  &  &  \\ \cline{2-3}
\multicolumn{1}{|l|}{}  & \multicolumn{1}{l|}{\begin{tabular}[c]{@{}l@{}}$f_{\hat{s}}$: one fully connected layer\\ $(2k_M, 4k_M)$\end{tabular}}   & \multicolumn{1}{l|}{$4k_M$ (for $\vmu$ and $\ve{\sigma}$)}    &  &  &  &  \\ \cline{2-3}
\multicolumn{1}{|l|}{}  & \multicolumn{1}{l|}{\begin{tabular}[c]{@{}l@{}}$f_s$: Multi-head layer, for head $m$\\ $(2k_M, 2k_m)$\end{tabular}}  & \multicolumn{1}{l|}{$2k$}  &  &  &  &  \\ \cline{1-3}
\multicolumn{1}{|l|}{\begin{tabular}[c]{@{}l@{}}AdaSem \\ Decoder\end{tabular}}   & \multicolumn{1}{l|}{\begin{tabular}[c]{@{}l@{}}$f_d$: 4 fully connected layers \\ $(2k_M,  512, 128,  32, 6)$\end{tabular}}       & \multicolumn{1}{l|}{6}  &  &  &  &  \\ \cline{1-3} &  & &  &  &  & 
\end{tabular}
\vspace{-2ex}
\end{table}
\subsection{Performance comparison}
To compare to the baseline, we also assume a dedicated Wi-Fi for AdaSem. A 20 MHz 802.11n/ac Wi-Fi consists of 64 subcarriers, while 52 of them are used to carry data \cite{hiertz2010ieee}. For a fair comparison, we only assume that half of them are used for the uplink. The legacy symbol duration of 802.11ac is 3.2 $\mu$s with a guard interval of $0.8$ $\mu$s, thus we estimate the communication time of transmitting $k$ symbols to be $k\cdot 4\cdot 10^{-6}/26$ second, i.e., $k\cdot 1.54\cdot 10^{-4}$ ms.

\begin{table*}[t]
\caption{Performance Comparison}
\centering
\label{Tab:comparison}
\begin{tabular}{|l|l|l|l|l|l|l|}
\hline
 \multirow{2}{*}{Scheme} & \multirow{2}{*}{Avg. $k^\ast$} & Position MAE & Orientation Error & Encoder latency  & Decoder latency  & E2E latency \\
  &  & (in m) &  (in °) & (in ms) &  (in ms) & (in ms)
\\ \hline
Baseline & N/A               & 0.022            & 2.523                & {\begin{tabular}[c]{@{}l@{}} 14.670 (compression +  \\ coding + transmission) \end{tabular}}                      & {\begin{tabular}[c]{@{}l@{}}  3.174 (decompression + decoding) \\ 2.628 (regression)\end{tabular}}                      & 20.367           \\ \hline
AdaSem (fix $k=64$)      &    64            & 0.011            & 0.497                 & {\bf 1.753}                & {\bf 0.437}               & {\bf 2.288}        \\\hline
AdaSem (fix $k=128$)     &    128           & 0.013            & 0.546                & 1.759               & 0.480                & 2.436          \\\hline
AdaSem (fix $k=256$)     & 256               & 0.008            & 0.341                 & 3.190               & 0.466               & 4.050           \\ \hline
AdaSem (fix $k=512$)    & 512              & {\bf 0.007}            & {\bf 0.299}                & 5.129                & 0.461              & 5.338           \\\hline
AdaSem (adaptive $k$)   &   $244$            &    $0.008$        &   $0.322$           &        $4.082$        &     $0.468$         &         $4.926$   \\\hline
\end{tabular}
\vspace{-0.2cm}
\end{table*}

The performance comparisons between the baseline solution, AdaSem with fixed $k$ (i.e., without the adaptation policy $\pi$), and AdaSem are shown in Table \ref{Tab:comparison}. Note that the results shown are averaged over the test dataset. Moreover, Figure \ref{fig:delay_pos} and \ref{fig:delay_arg} show the tradeoff between the end-to-end latency and the position \ac{MAE} and angular distance, respectively. The position \ac{MAE} is defined as $1/N \sum_{i=1}^N|\vp_i - \hat{\vp}_i|$ for a test dataset size $N$.
The observations are summarized below.
\begin{itemize}
    \item {\bf Semantic advantage.} In Table \ref{Tab:comparison}, we show that both semantic schemes (with fixed or adaptive $k$) dramatically outperform the baseline in terms of both end-to-end latency and inference accuracy. Even with the largest $k=512$, the semantic scheme achieves $73.8\%$ lower delay, $68.2\%$ lower position error, and $88.1\%$ lower angular error than the baseline. The proposed AdaSem approach achieves $75.8\%$ lower delay, $63.6\%$ lower position error, and $87.2\%$ lower angular error, with an additional advantage of more efficient channel utilization.
    \item {\bf Delay-Accuracy tradeoff.} In Figure \ref{fig:perf_comp}, we average the delay of AdaSem for different intervals and compute the corresponding average inference error. We compare its performance against the averaged performance of the baseline and AdaSem with fixed $k$. We show that AdaSem achieves a good tradeoff between the end-to-end delay and the inference error.  
    \item {\bf Flexible resource utilization.} Compared to fixed $k$, although AdaSem sacrifices slightly the inference latency due to the additional complexity of the variational encoder $f_{\hat{s}}$, it enables flexible and efficient usage of the radio resource (shown in Figure \ref{fig:adaptive_k}), and is more robust against varying channel conditions (shown in Figure \ref{fig:angerr_vs_snr}).
\end{itemize}

\begin{figure}[t]
     \centering
     \begin{subfigure}{.43\textwidth}
         \centering         \includegraphics[width=\textwidth]{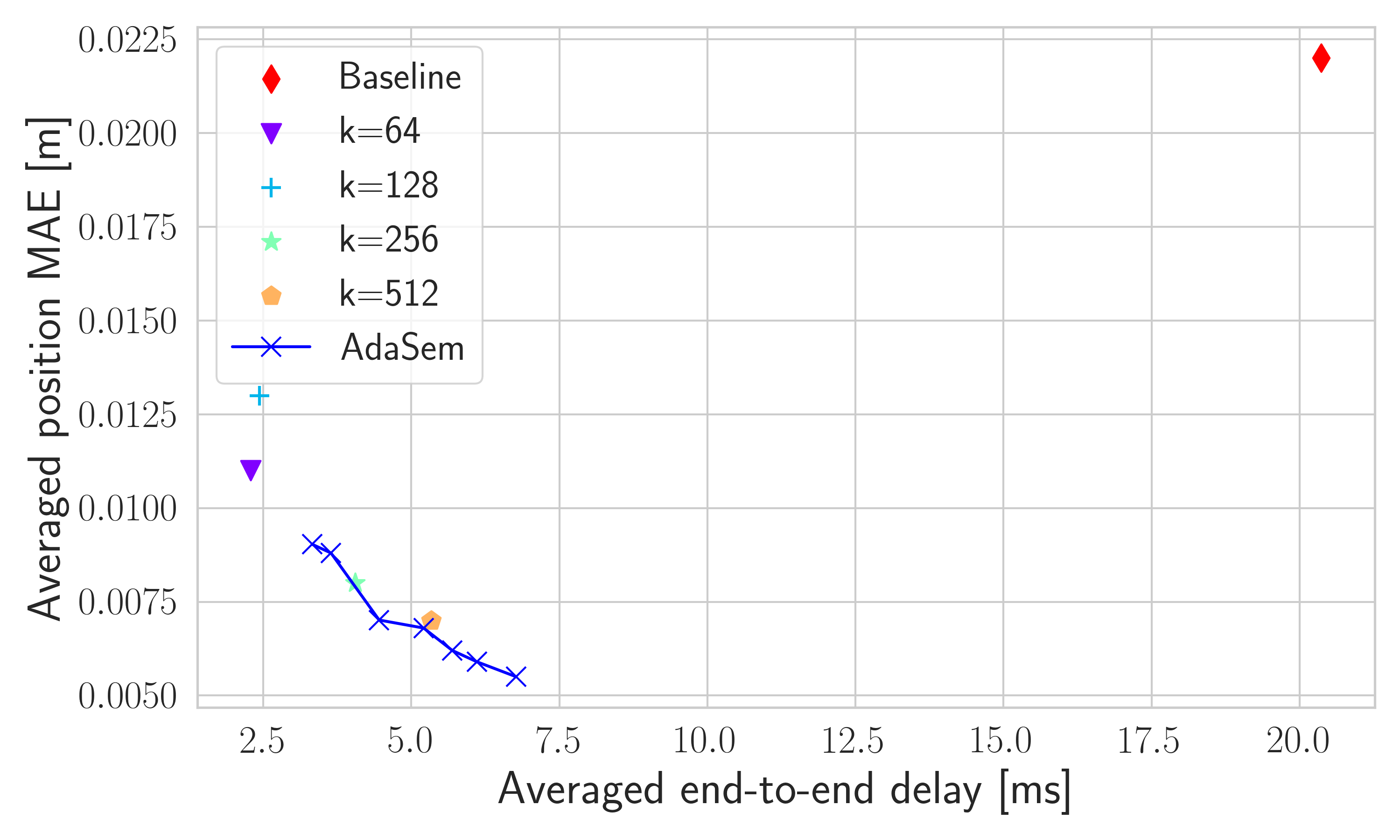}
         \caption{End-to-end delay versus position \ac{MAE}.}
         \label{fig:delay_pos}
     \end{subfigure}
     \begin{subfigure}{0.43\textwidth}
         \centering
    \includegraphics[width=\textwidth]{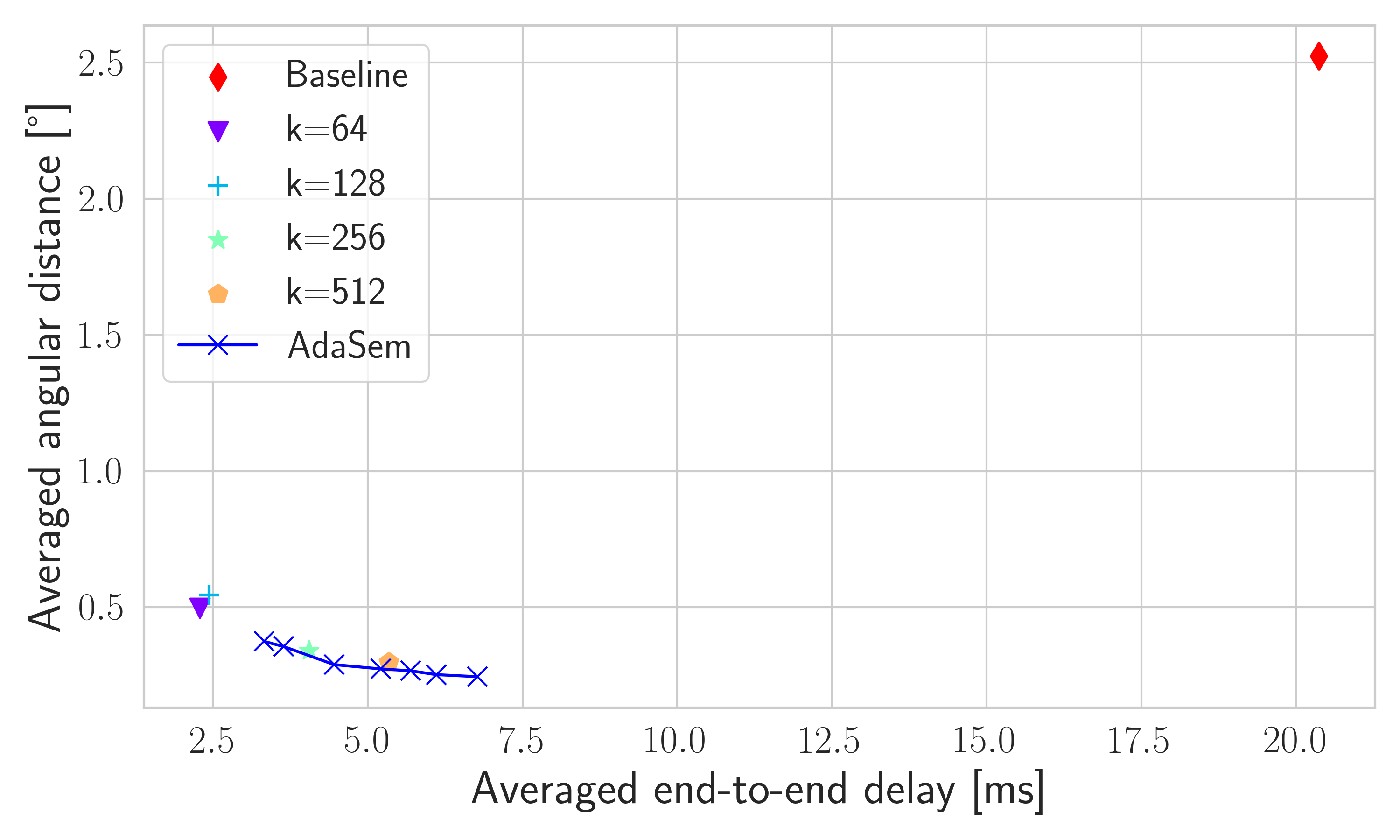}
         \caption{End-to-end delay versus angular distance.}
         \label{fig:delay_arg}
     \end{subfigure}
        \caption{Performance comparison}
        \label{fig:perf_comp}
\end{figure}

\subsection{Interpretation of the extracted features and encoded symbols}
To investigate the type of features that are learned for our camera relocalization task, in Figure \ref{fig:feat_interpretation} we plot the output of the visual feature extractor for two examples: when $k=512$ and $k=64$, respectively. As shown in Table \ref{Tab:architecture}, the output of the visual feature extractor contains $24$ feature maps with dimension $28\times28$. It can be seen that the learned features represent the edges and corners of the image, confirming our assumption that the lower layers of an object detector \ac{DNN} are trained to extract such types of features. These features are also similar to those used in classical geometric-based localization algorithms that depend on keypoint matching, as well as those of keypoint detectors that detect edges and corners, e.g., the \ac{ORB} \cite{rublee2011orb} detector. We also observe that with a lower value of $k$, the features are less detailed, due to the lower communication rate.

In Figure \ref{fig:symbol_interpretation} we show $\vs$ for $k = 64$ and $512$ respectively using a particular input sample containing the image in Figure \ref{fig:img}. We observe that in this case, $k=512$ leads to a much higher level of redundancy since the encoded symbols are highly correlated, while for $k=64$, the symbols are less correlated. With this observation, we hypothesize that, unlike the conventional communications schemes, where the transmitted information (e.g., features) are encoded onto channel symbols independently, the semantic autoencoder learns a joint distribution of the meaningful features and may find a way to encode their relationships positionally within $\vs$, when given sufficient degrees of freedom. This shows that, the reason that the visual-based localization tasks can benefit from semantic communications approaches is partially due to the ability of the encoder to spatially correlate the detected features in the transmitted symbols.

\begin{figure*}[t]
     \centering
     \begin{subfigure}{.232\textwidth}
\centering\includegraphics[width=.9\textwidth]{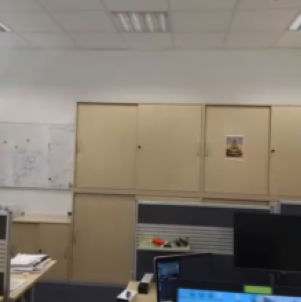}
         \caption{Input image $\mI\in\R^{3\times 244\times 244}$.}
         \label{fig:img}
     \end{subfigure}
     \begin{subfigure}{0.35\textwidth}
         \centering\includegraphics[width=.9\textwidth]{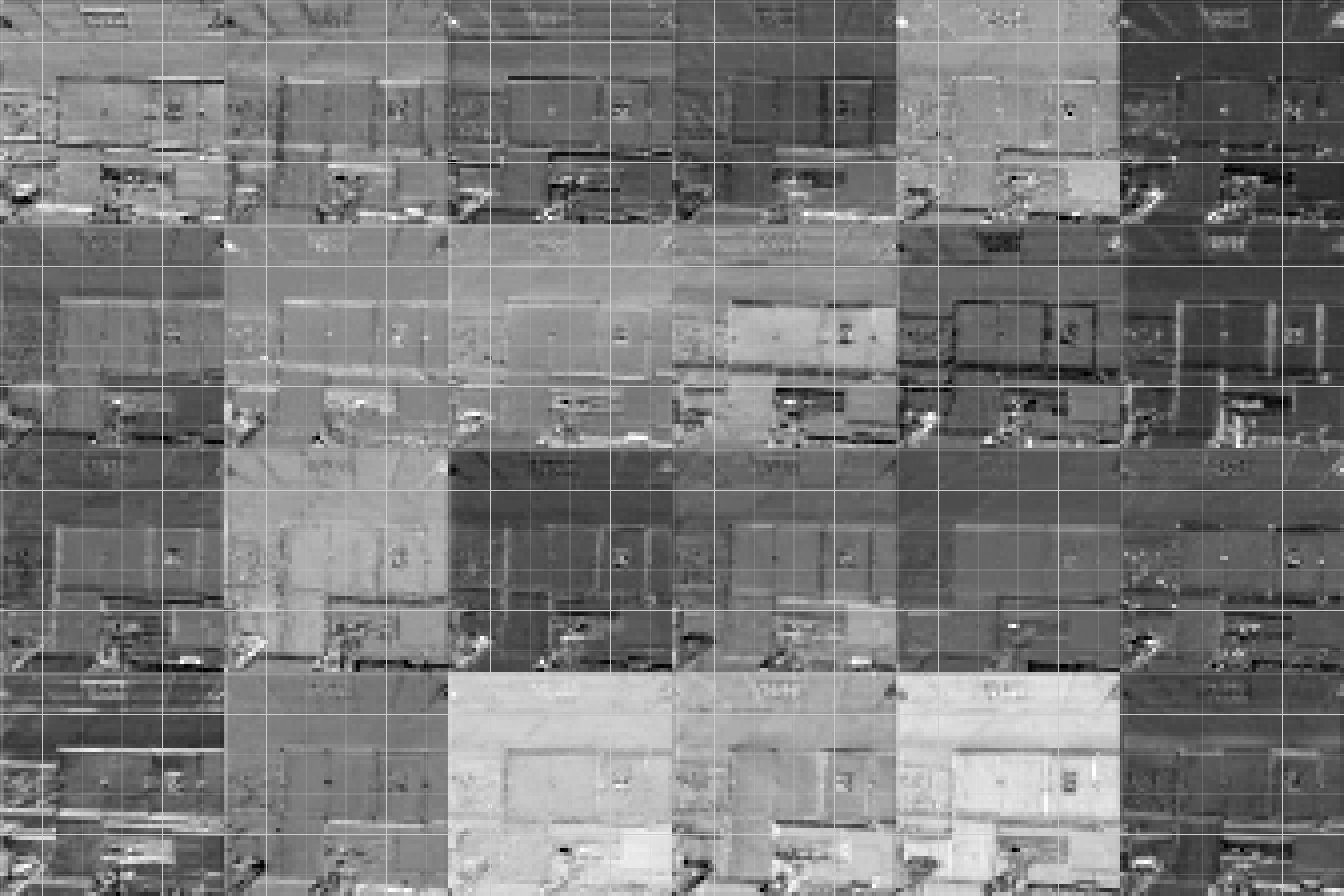}
         \caption{Extracted visual features: $k=512$.}
         \label{fig:feature_512}
     \end{subfigure}
     \begin{subfigure}{.35\textwidth}
         \centering\includegraphics[width=.9\textwidth]{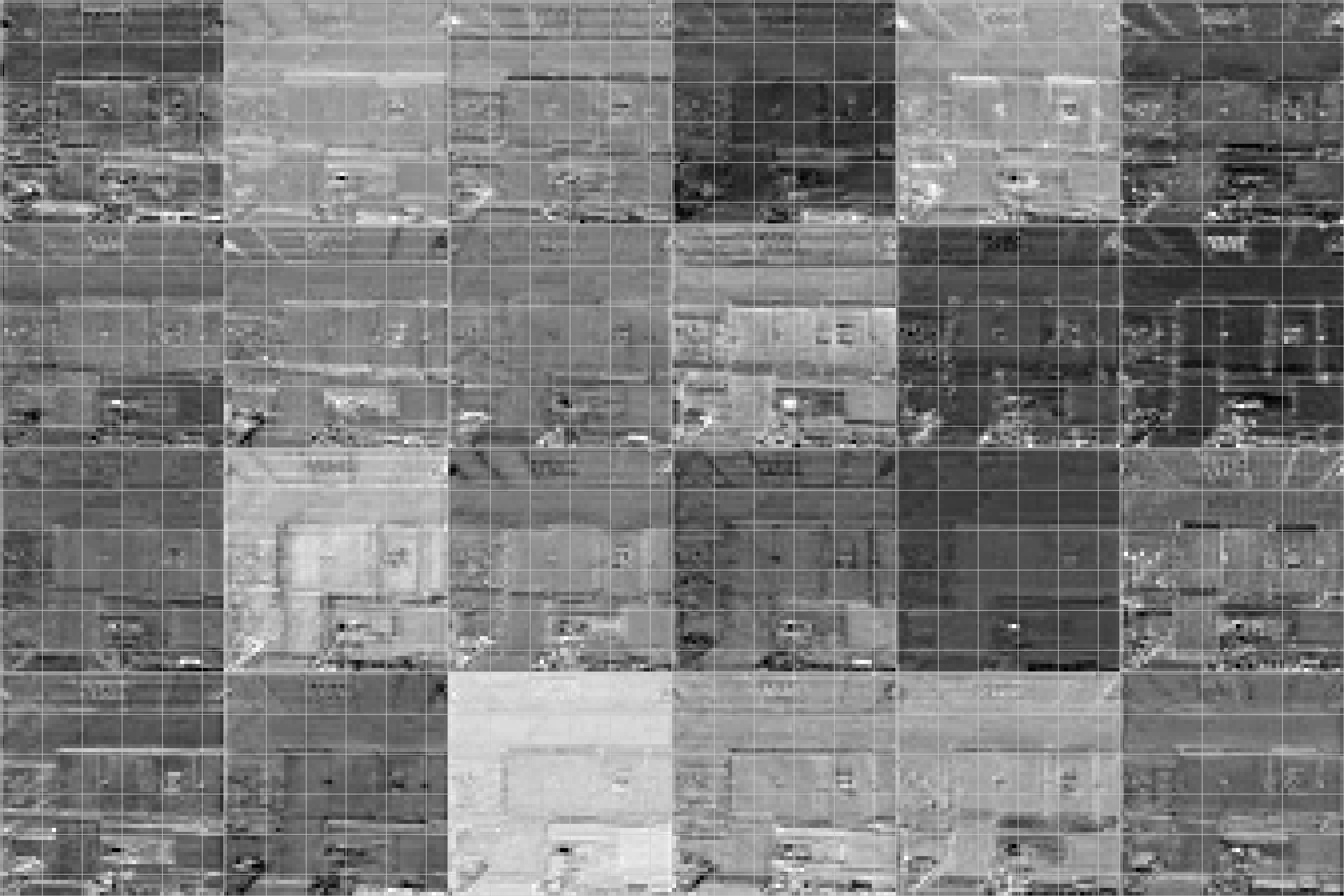}
         \caption{Extracted visual features: $k=64$.}
         \label{fig:feature_2}
     \end{subfigure}
        \caption{Feature interpretation.}
        \vspace{-2ex}
        \label{fig:feat_interpretation}
\end{figure*}

\begin{figure}[t]
     \centering
     \begin{subfigure}{.24\textwidth}
         \centering
         \includegraphics[width=\textwidth]{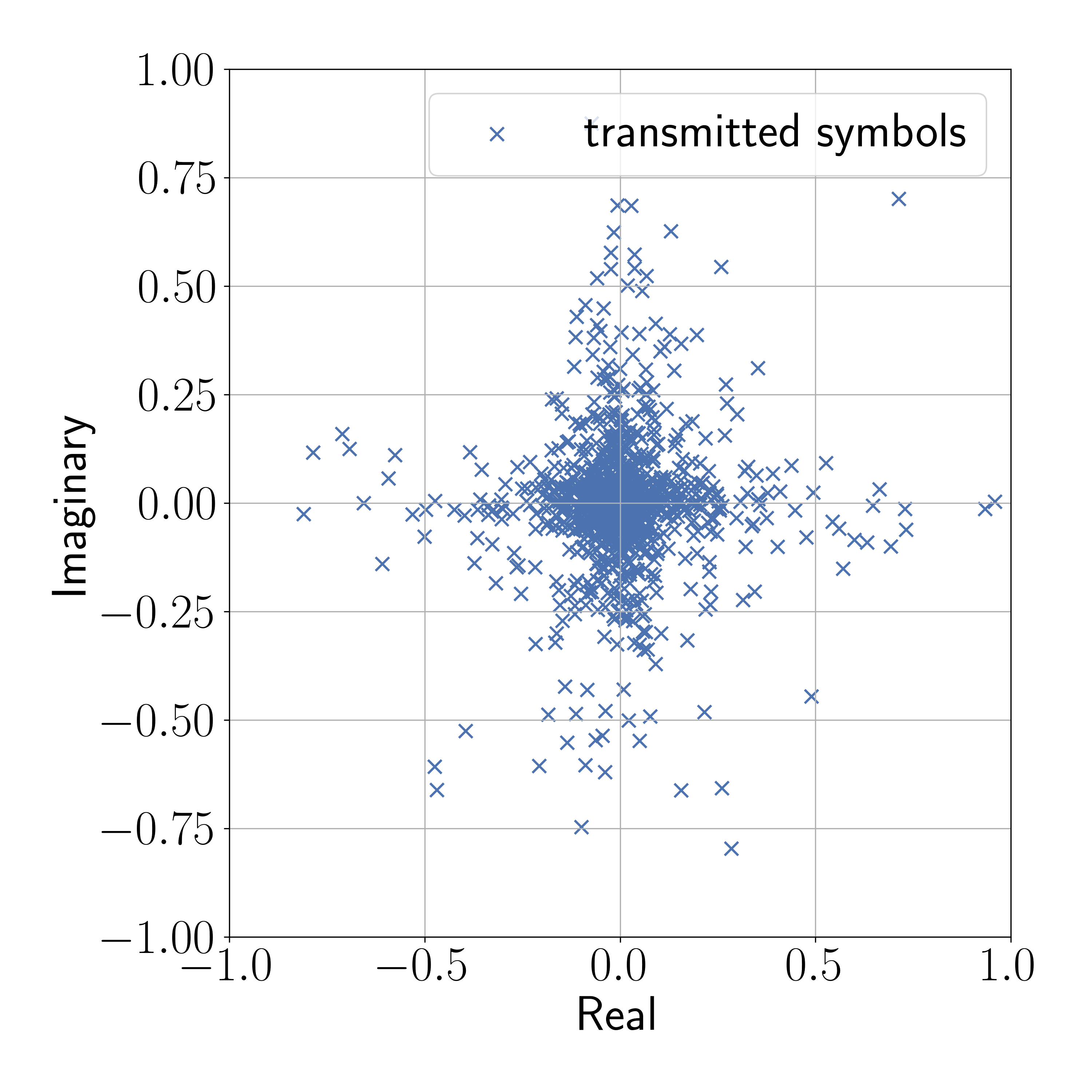}
         \caption{Encoded symbols: $k=512$.}
         \label{fig:sym_k512}
     \end{subfigure}
     \begin{subfigure}{0.24\textwidth}
         \centering
         \includegraphics[width=\textwidth]{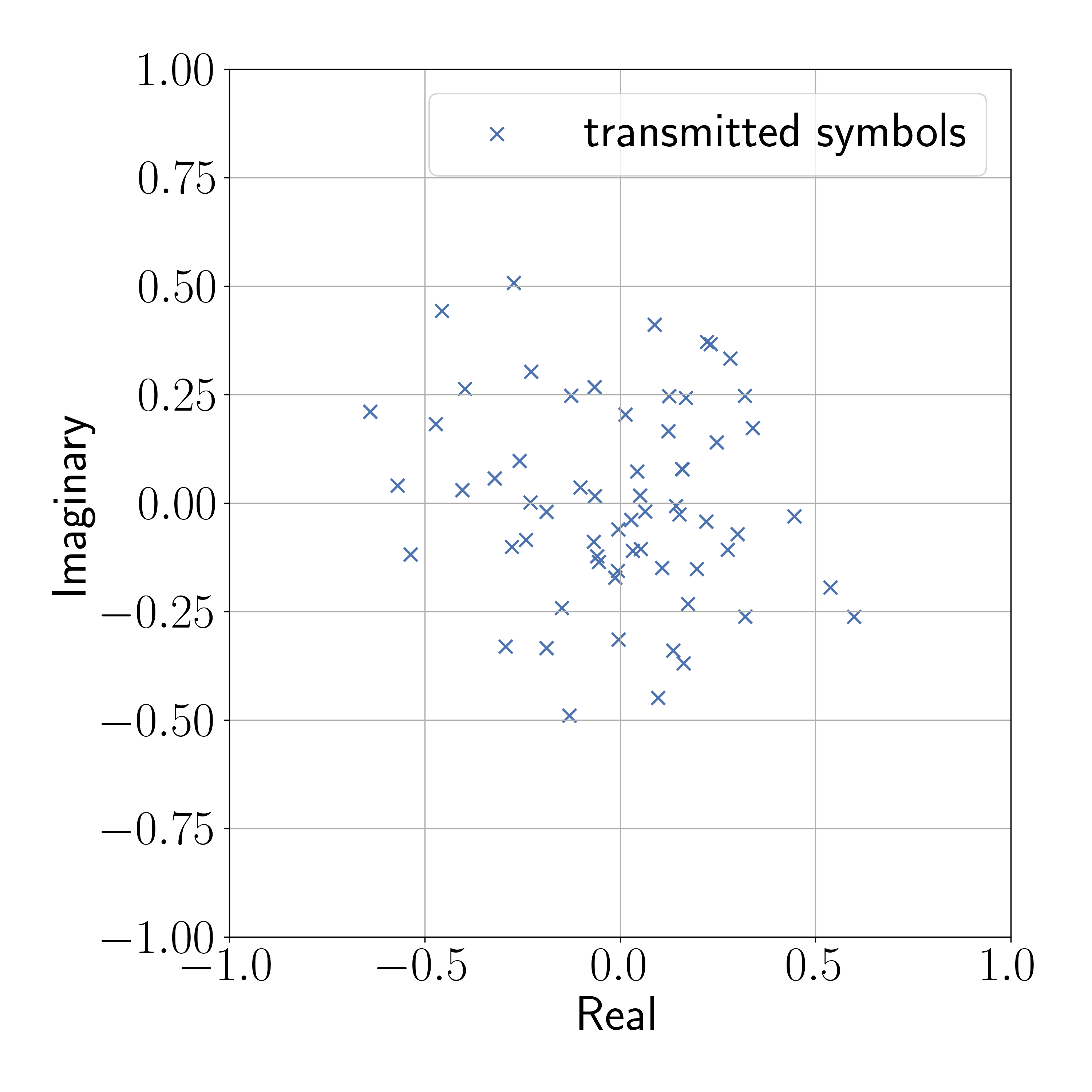}
         \caption{Encoded symbols: $k=64$.}         \label{fig:sym_k64}
     \end{subfigure}
        \caption{Encoded symbols.}   
        \vspace{-2ex}
        \label{fig:symbol_interpretation}
\end{figure}

\subsection{Evaluation on adaptation policy}
In Figure \ref{fig:adaptive_k} we show how the policy $\pi$ adapts $k$ to the channel SNR. In general, when the channel quality improves, the policy reduces the dimension of the encoded symbols, as less redundancy is required to achieve the same rate-distortion tradeoff. It is worth noting that $k$ depends not only on the channel condition but also on the source data distribution since $k$ is optimized based on the upper bound of $I(\hat{S}; X,C)$. Then, Figure \ref{fig:angerr_vs_snr} shows that AdaSem provides stable, low-error orientation estimation under varying averaged \ac{SNR} values. A similar result is also derived for position estimation, but due to the limited space, we omit the figure.  
Lastly, we show in Figure \ref{fig:r_vs_d}, that the tradeoff between the estimated communication rate and distortion is convex, as expected, based on the empirical computation of Equation \eqref{eqn:empirical_vib} in the testing phase.
\begin{figure}[t]
     \centering
\includegraphics[width=0.47\textwidth]{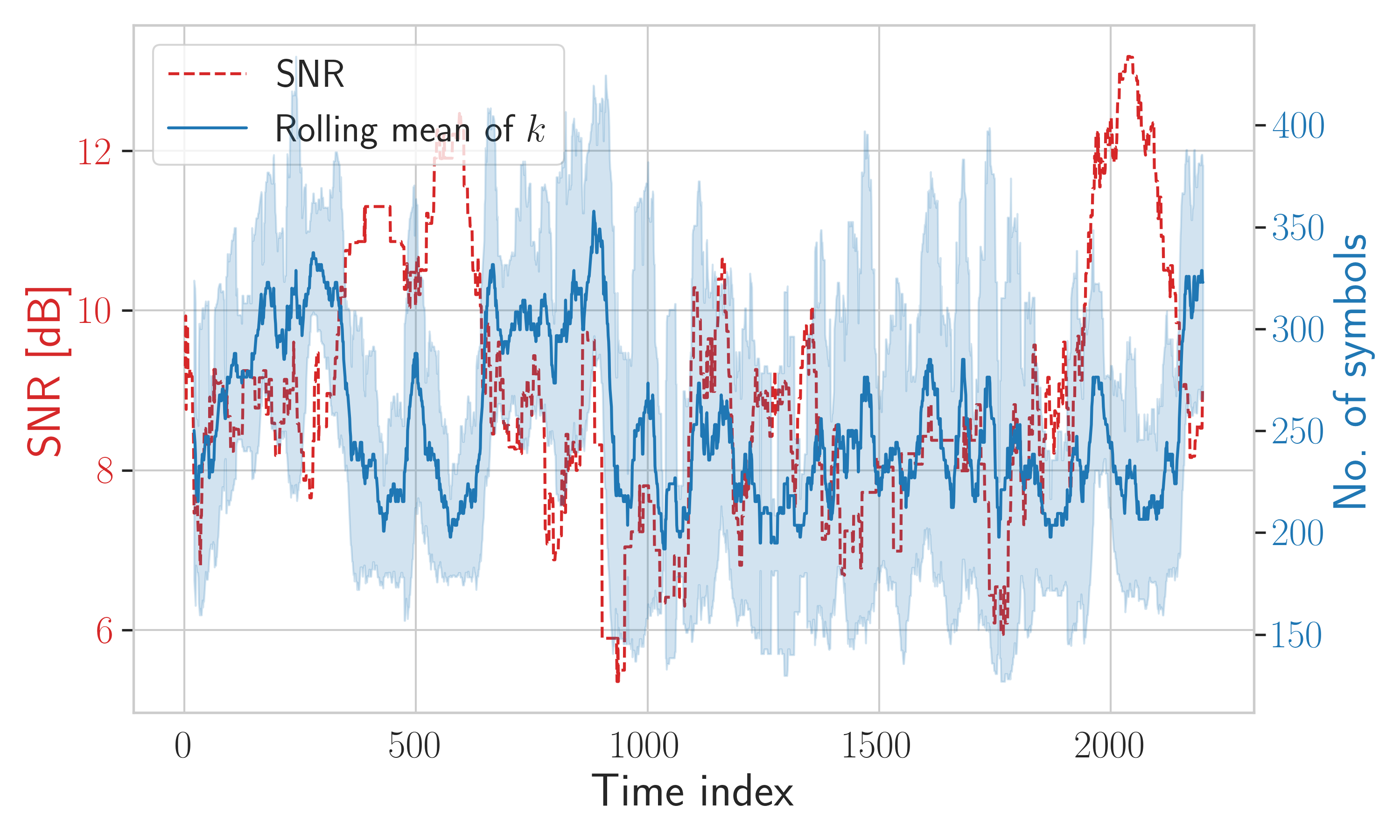}
\vspace{-1ex}
      \caption{Channel-aware number of symbols: rolling mean of the adaptive $k$ along the testing trajectory.}
\label{fig:adaptive_k}
\end{figure}
\begin{figure}[t]
     \centering
\includegraphics[width=0.45\textwidth]{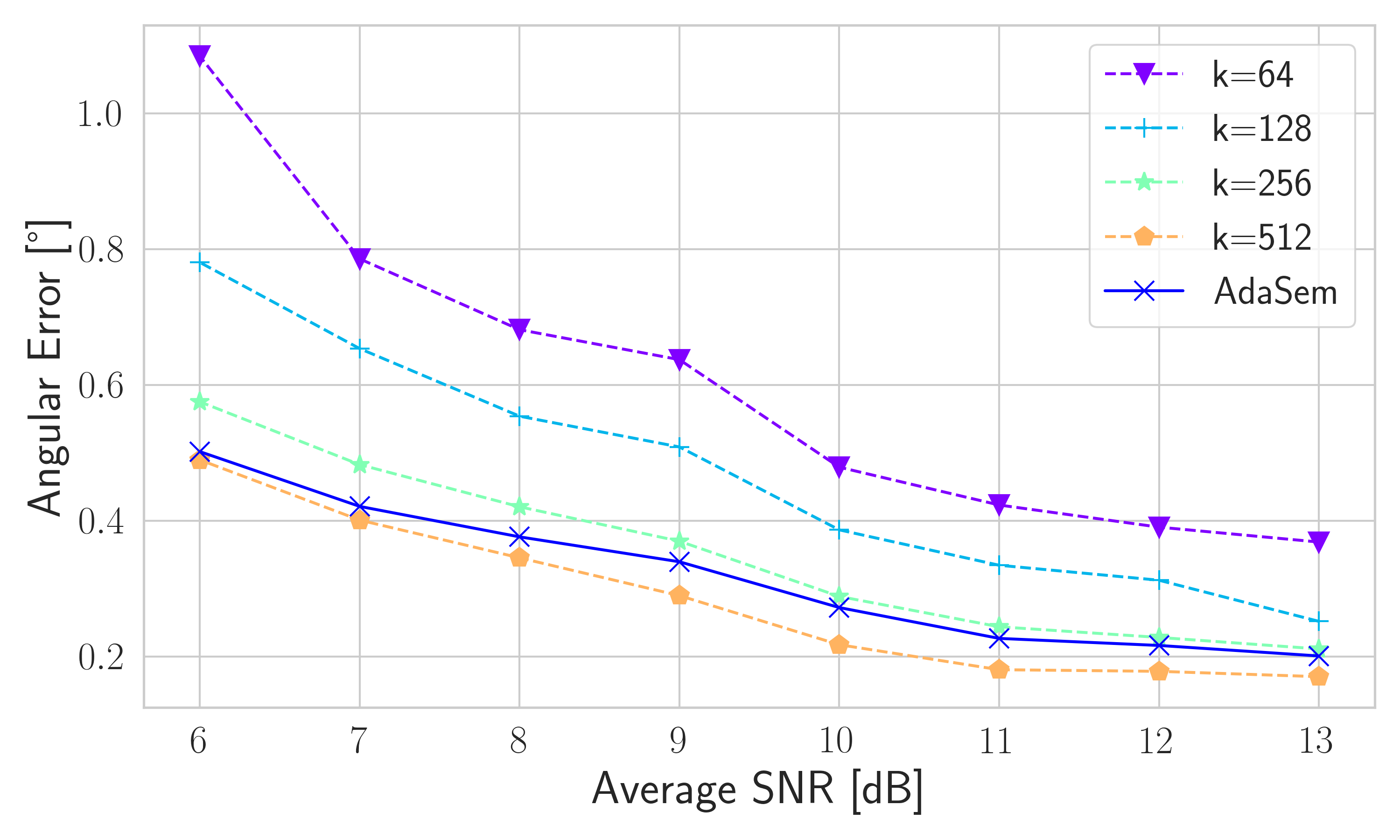}
\vspace{-1ex}
      \caption{Averaged angular distance versus \ac{SNR}.}
\label{fig:angerr_vs_snr}
\end{figure}
\begin{figure}[t]
     \centering
\includegraphics[width=0.45\textwidth]{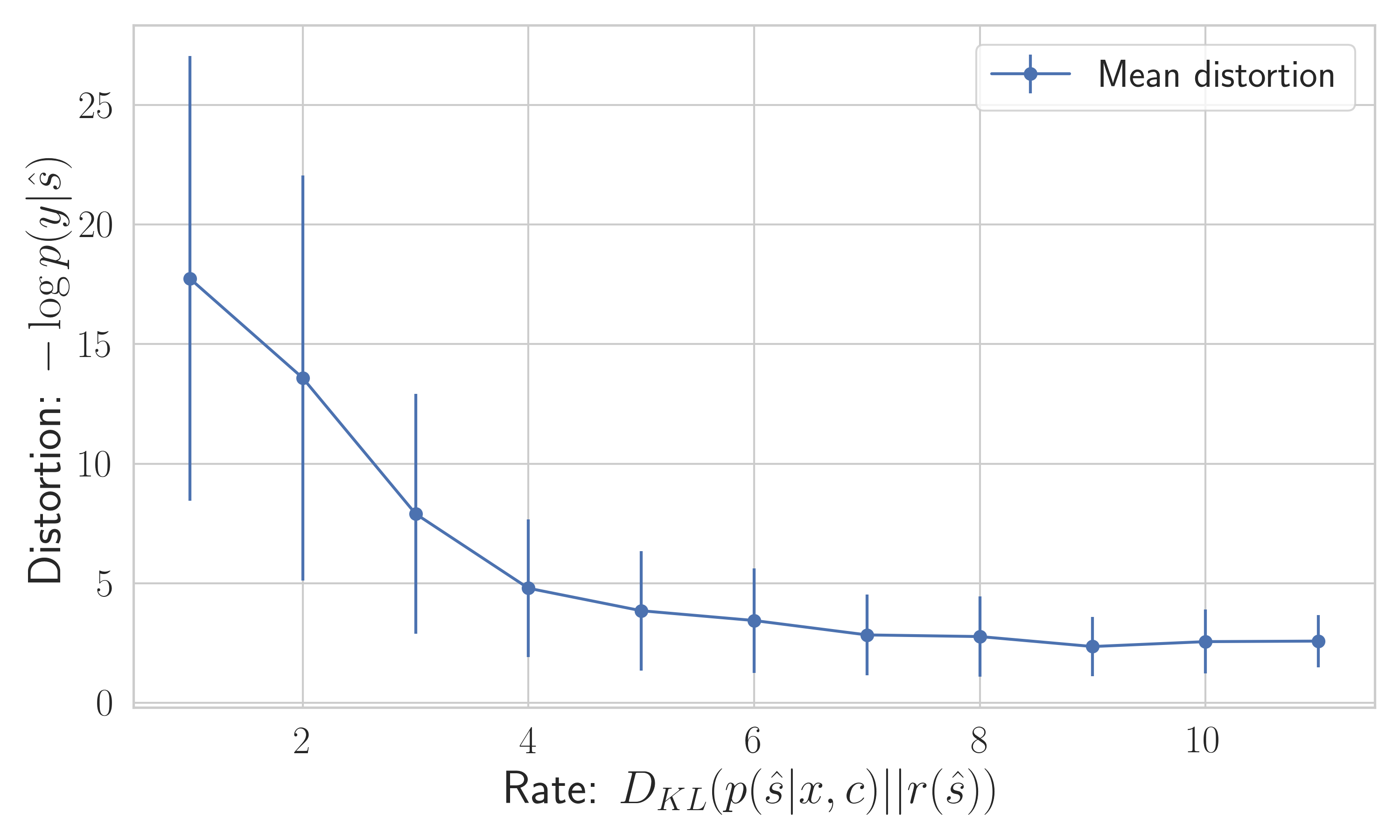}
\vspace{-1ex}
         \caption{Rate versus distortion.}
\label{fig:r_vs_d}
\end{figure}
\section{Conclusion}\label{sec:Conclusion}
We have developed a goal-oriented semantic communications system for camera relocalization, called AdaSem, that can trade off end-to-end latency against inference performance.
Compared to prior works, AdaSem not only improves the transmission delay by utilizing an adaptive policy for varying the communication rate but also uses a lightweight architecture for the encoder and decoder, such that the end-to-end latency target is met.
To optimize the adaptive communication rate policy, we define a channel-aware information bottleneck and derive a variational upper bound, which is used to optimize the communication rate and inference distortion tradeoff.
The result is an extremely fast, end-to-end camera relocalization system that not only improves the estimation performance but also the end-to-end latency when compared to a real implemented baseline service that uses separate source and channel coding for camera pose estimation.

\appendix
\subsection{Proof of Corollary \ref{corol:VIB_compute}}\label{sec:coro_proof}
\begin{proof}
We first derive the lower bound for $I(\hat{S}; Y)$. Since $\KL(p(\vy|\hat{\vs})||q_{\phi}(\vy|\hat{\vs}))\geq 0$, we have 
\begin{align*}
I(\hat{S}; Y) &  = \int p(\vy, \hat{\vs})\log \frac{p(\vy|\hat{\vs})}{p(\vy)}d\vy  d\hat{\vs} \\
 \geq  & \int p(\vy, \hat{\vs})\log q_{\phi}(\vy|\hat{\vs})d\vy  d\hat{\vs} - \int p(\vy, \hat{\vs})\log p(\vy)d\vy \, d\hat{\vs}. 
\end{align*}
To compute the first term, we use $p(\vy, \hat{\vs}) = \int p(\vy, \hat{\vs}, \vx, c) d\vx \, dc$, and with the Markov chain assumption, we further have  $p(\vy, \hat{\vs}, \vx, c) = p(\hat{\vs}|\vx, c)p(\vy|\vx, c) p(\vx)p(c) = p(\hat{\vs}|\vx, c)p(\vx, \vy, c)$. Thus, by replacing $p(\vy, \hat{\vs}) $ and using a variational approximation $p_{\theta}(\hat{\vs}|\vx, c)$,  the first term is equivalent to $\Ex_{p(\vx, \vy, \vz)}\Ex_{p_{\theta}(\hat{\vs}|\vx, c)} [\log q_{\phi}(\vy|\hat{\vs})]$. The second term is simply the entropy $H(Y)$ since $\int p(\vy|\hat{\vs})p(\hat{\vs})d\hat{\vs}=p(\vy)$. Note that since $H(Y)$ is a constant and independent of our optimization procedure, we have 
\begin{equation*}
- I(\hat{S}; Y) \leq - \Ex_{p(\vx, \vy, c)}\Ex_{p_{\theta}(\hat{\vs}|\vx, c)} \left[\log q_{\phi}(\vy|\hat{\vs})\right].
\end{equation*}
Now we derive the upper bound for $I(\hat{S}; X, C)$. Similarly, since $\KL(p(\hat{\vs})||r(\hat{\vs}))\geq 0$, and with the variational approximations $p_{\theta}(\hat{\vs}|\vx, c)$ and $r(\hat{\vs})$, we have

\begin{align*}
I(\hat{S}; X, C) & = \int p(\hat{\vs}|\vx, c) p(\vx, c)\log \frac{p(\hat{\vs}|\vx, c)}{p(\hat{\vs})}d\hat{\vs} \, d\vx \, dc \\
& \leq \int p_{\theta}(\hat{\vs}|\vx, c) p(\vx, c)\log \frac{p_{\theta}(\hat{\vs}|\vx, c)}{r(\hat{\vs})}d\hat{\vs} \, d\vx \, dc\\
& = \Ex_{p(\vx, c)} \KL\big(p_{\theta}(\hat{\vs}|\vx, c)||r(\hat{\vs})\big).
\end{align*}
Using the Markov chain assumption in Figure \ref{fig:MarkovAssumption}, we can combine the upper bound for $- I(\hat{S}; Y)$ and $I(\hat{S}; X, C)$, and derive \eqref{eqn:vib_upperbound}. 
\end{proof}

\subsection{Computation of Angular Loss}\label{sec:agular}
The angle $\theta$ between two unit quaternions $\vq$ and $\vq'$ yields $\left[\sin{(\theta/2)}\vu, \cos{(\theta/2)}\right] = \vq\vq'^{-1}$, where $\vu\in\R^3$ is the unit vector along the axis of the minimal rotation between the two quaternions, and $\vq\vq'^{-1}$ is the quaternion product of $\vq$ and the inverse of $\vq'$ (for details of quaternion algebraic properties, please refer to \cite{altmann2005rotations}).  Let $\bar{\vq}:=\vq\vq'^{-1}$, then $\sin(\theta/2)$ is equivalent to the norm of its vector part $\vv(\bar{\vq})$ and $\cos(\theta/2)$ is equivalent to the absolute scalar part $s(\bar{\vq})$. Since the angular distance is within $[-\pi, \pi]$, minimizing $\theta$ in this region is equivalent to maximizing $\cos(\theta/2)$. Thus, we minimize $-|s(\bar{\vq})|$ to minimize $\theta$. Since $s(\bar{\vq})=\vq\cdot\vq'\in\R$, the angular loss can be defined as $-|\vq\cdot\vq'|$. 

\acrodef{3GPP}{3rd generation partnership project}
\acrodef{5G}{fifth generation}
\acrodef{6G}{sixth generation}
\acrodef{AI}{artificial intelligence}
\acrodef{AR}{augmented reality}
\acrodef{AWGN}{additive white Gaussian noise}
\acrodef{BS}{base station}
\acrodef{BCQ}{batch-constrained deep Q-learning}	
\acrodef{CQI}{channel quality indicator}
\acrodef{CIO}{cell individual offset}
\acrodef{CDF}{cumulative density function}
\acrodef{CNN}{convolutional neural network}
\acrodef{CVAE}{conditional variational autoencoder}
\acrodef{CVIB}{conditional variational information bottleneck}
\acrodef{CSI}{channel state information}
\acrodef{DEEN}{deep energy estimator network}
\acrodef{DNN}{deep neural network}
\acrodef{DoF}{degree-of-freedom}
\acrodef{DDPG}{deep deterministic policy gradient}
\acrodef{DQN}{deep Q-learning}
\acrodef{DRL}{deep reinforcement learning}
\acrodef{EC}{effective communication}
\acrodef{eICIC}{enhanced inter-cell interference coordination}
\acrodef{eMBB}{enhanced mobile broadband}
\acrodef{ELBO}{evicence lower bound}
\acrodef{FLOP}{Floating Point Operations}

\acrodef{GP}{Gaussian process}

\acrodef{HO}{handover}
\acrodef{HOM}{handover margin}
\acrodef{HOL}{Too-late-handover}
\acrodef{HOE}{Too-early-handover}
\acrodef{HOW}{Wrong-cell-handover}
\acrodef{HOPP}{Ping-pong-handover}
\acrodef{HFR}{handover failure ratio}
\acrodef{HetNet}{heterogeneous network}

\acrodef{IoT}{internet of things}
\acrodef{ICI}{inter-cell interference}
\acrodef{IMU}{inertial measurement unit}
\acrodef{JSCC}{joint source and channel coding}
\acrodef{KL}{Kullback-Leibler}
\acrodef{KPI}{Key Performance Indicator}
\acrodef{LTE}{long-term evolution} 
\acrodef{LSL}{latency service level}

\acrodef{MAE}{mean absolute error}
\acrodef{MDP}{Markov decision process}
\acrodef{MRO}{mobility robustness optimization}
\acrodef{MLB}{mobility load balancing}
\acrodef{MLP}{multi-layer perception}

\acrodef{NLES}{nonlinear equation system}
\acrodef{NTSCC}{nonlinear transform source-channel coding}
\acrodef{NTC}{nonlinear transform coding}
\acrodef{OFDM}{orthogonal frequency division multiplexing}
\acrodef{ORB}{Oriented FAST and Rotated BRIEF}
	
\acrodef{PDF}{probability density function}

\acrodef{PPR}{ping-pong ratio}
    
\acrodef{QoE}{quality of experience}
\acrodef{QoS}{quality of service}

\acrodef{RAN}{radio access network}
\acrodef{RE}{resource element}
\acrodef{RRM}{radio resource management}
\acrodef{RLF}{radio link failure}
\acrodef{RSRP}{reference signal received power}
\acrodef{RANSAC}{random sample consensus}
\acrodef{RTP}{real-time transport protocol}
\acrodef{SAC}{soft actor-critic}
\acrodef{SNR}{signal-to-noise ratio}
\acrodef{SINR}{signal-to-interference-plus-noise ratio}
\acrodef{SIR}{signal-to-interference ratio}
\acrodef{SON}{self-organizing network}
\acrodef{SVM}{support vector machine}
\acrodef{SVD}{singular value decomposition}
\acrodef{SLAM}{simultaneous localization and mapping}
 \acrodef{TTT}{time-to-trigger} 
\acrodef{TD3}{twin delayed deep deterministic policy gradient}
\acrodef{TSL}{throughput service level}

\acrodef{UE}{user equipment}
\acrodef{UI}{user interface}
\acrodef{URLLC}{ultra-reliable and low-latency communication}
\acrodef{VAE}{variational autoencoder}
\acrodef{VIB}{variational information bottleneck}
\acrodef{VR}{virtual reality}
\acrodef{VO}{virtual odometry}




\bibliographystyle{IEEEtran}
\bibliography{myreferences}

\end{document}